\documentclass{statsoc}
\usepackage[a4paper]{geometry}
\usepackage{amsmath, amssymb, bm,enumerate, mathrsfs,mathtools, enumitem}
\usepackage{latexsym,color,verbatim, array, multirow}
\usepackage{multirow}
\usepackage{array, tabularx}

\usepackage{etoolbox}
\makeatletter
\patchcmd{\@makecaption}
  {\parbox}
  {\advance\@tempdima-\fontdimen2} 
  {}{}
\makeatother  
\usepackage[caption = false]{subfig} 
\usepackage{enumerate}
\usepackage{comment}
\usepackage{natbib}

\usepackage{eufrak}
\usepackage{float}
\usepackage{graphics,amssymb,color}
\usepackage{graphicx}
\usepackage{grffile}
\usepackage{mathtools}
\usepackage{algorithm}
\usepackage{algorithm,algorithmic}
\usepackage{xfrac}

\usepackage{natbib}

\usepackage{hyperref}
\hypersetup{colorlinks,
		citecolor=blue,
		linkcolor=blue,
		urlcolor=black}

\newcommand{\real}{\mathbb{R}}

\newtheorem{theorem}{Theorem}
\newtheorem{lemma}{Lemma}

\newtheorem{remark}{Remark}

\newcommand*{\Scale}[2][4]{\scalebox{#1}{$#2$}}

\makeatletter
\renewcommand*\env@matrix[1][\arraystretch]{%
  \edef\arraystretch{#1}%
  \hskip -\arraycolsep
  \let\@ifnextchar\new@ifnextchar
  \array{*\c@MaxMatrixCols c}}
\makeatother

\usepackage{makecell}

\newcommand\MyBox[2]{
  \fbox{\lower0.75cm
    \vbox to .6cm{\vfil
      \hbox to 1cm{\hfil\parbox{0.8cm}{#1\\#2}\hfil}
      \vfil}%
  }%
}

\newcolumntype{M}[1]{>{\centering\arraybackslash}m{#1}}
\newcolumntype{N}{@{}m{0pt}@{}}

\title[Effect size estimation]{Selection-adjusted inference: an application to confidence intervals for $cis$-eQTL effect sizes}

\author[Panigrahi {\it et al.}]{Snigdha Panigrahi}
\address{Department of Statistics,
		 Stanford University,
         CA, USA.}
\email{snigdha@stanford.edu}

\author{Junjie Zhu}
\address{Department of Electrical Engineering,
	      Stanford University,
         CA, USA.}
\email{jjzhu@stanford.edu}

\author[Panigrahi et al.]{Chiara Sabatti}
\address{Department of Biomedical Data Science and  Department of Statistics,
		 Stanford University,
         CA, USA.}
\email{sabatti@stanford.edu}

\begin{document}

\begin{abstract}
The goal of eQTL studies is to identify the genetic variants that influence the expression levels of the genes in an organism. High throughput technology has made such studies possible: in a given tissue sample, it enables us to quantify the expression levels of approximately 20,000 genes and to record the  alleles present at millions of genetic polymorphisms. While obtaining this data is relatively cheap once a specimen is at hand, obtaining human tissue remains a costly endeavor:  eQTL studies continue to be based on relatively small sample sizes, with this limitation particularly  serious for tissues as brain, liver, etc.---often the organs of most immediate medical relevance. 

Given the high dimensional nature of these datasets and the large number of hypotheses tested, the scientific community has adopted early on multiplicity adjustment procedures. These testing procedures  primarily control the  false discoveries rate for the identification of genetic variants with influence on the expression levels. In contrast, a problem that has not received much attention to date is that of providing estimates of the effect sizes associated with these variants, in a way that accounts for the considerable amount of selection. Yet, given the difficulty of procuring additional samples, this challenge is of practical importance.

We illustrate in this work how the recently developed  conditional inference approach can be deployed to obtain confidence intervals for the eQTL effect sizes with reliable coverage. In addition to interval estimates, we also provide a point estimate that approximately counters the effect of selection bias to calibrate the strength of discovered associations. The procedure we propose is based on a randomized hierarchical strategy. Such a strategy has a two-fold contribution: one, it reflects the selection steps typically adopted in state of the art investigations and two, it introduces the use of randomness instead of data splitting to maximize the use of available data. Analysis of the GTEx Liver dataset (v6) suggests that naively obtained confidence intervals would likely not cover the true values of effect sizes and that the number of local genetic polymorphisms influencing the  expression level of genes might be underestimated.
\end{abstract}
\maketitle

\section{Introduction}
\label{introduction}

The goal of an eQTL (expression quantitative trait loci) study is to identify  the genetic variants  that regulate the expression of genes in different biological contexts and quantify their effects. Using statistical terminology,  the outcome variables (typically on the order of 20,000) are molecular measurements of the gene expression and the predictors are genotypes  for single nucleotide polymorphisms (SNP), typically on the order of ~1,000,000). The variants that are discovered to regulate gene expression are referred as eVariants and careful estimation of their effect sizes is often deferred to follow-up studies.

One commonly studied sub-type of eVariants are those referred to as  {\it cis}-eQTL: the DNA variants in the neighborhood of a gene that influence its expression directly.  The majority of eQTL investigations have focused on detecting these {\it cis}-variants owing to their relative simple biological interpretation as well as the fact that restricting attention to this subset of gene and variant pairs reduces the number of tested hypotheses and leads to improved power. Still, even when concentrating on cis regulation, eQTL studies face a formidable multiplicity problem and also, a subsequent winner's curse during effect size estimation of discovered associations, with approximately 20,000 genes and an average of 7500 variants in each {\em cis} region.

Since the first studies \citep{brem2002genetic, schadt2003genetics, cheung2005mapping}, the eQTL research community has recognized the False Discovery Rate (FDR) as a relevant global error rate and adopted corresponding controlling strategies. As the density of SNP genotyping increased over time, it became apparent that naive application of FDR controlling strategies \citep{benjamini_fdr, storey_positive_2003} might lead to excessive false discoveries. Many correlated variants may all be correctly identified as associated with changes in gene expression, even when there is really only one causal effect. This abundance of  discoveries, all corresponding to one true signal, artificially increases the denominator of the false discovery proportion, leaving room for some additional false findings. To address this difficulty, more recent works 
\cite{gtex2015genotype, ongen2015fast, GTEx17} adopt a hierarchical strategy, which controls the FDR of eGene discoveries. Firstly, for each gene one tests the null hypothesis of no association with any local variants and the p-values corresponding to these tests are passed to an FDR controlling procedure.  Secondly, for those genes for which this null is rejected (eGenes), scientists proceed to try to identify which among the {\em cis} variants have an effect (eVariants). To this end both marginal testing \citep{BetS17} and multivariate regression models \cite{GTEx17} have been used to report the {\em cis} variants with significant association with eGenes.

Once eVariants have been detected, the logical next step is to attempt to estimate their effect sizes; and given the scarcity of biological samples, it is tempting to do so using the data at hand. However, since these discoveries have been selected out of a large number of possible associations, naive estimators based on the same data used for selection would result in inaccurate estimates. Indeed, this is a situation similar to that of GWAS, where this problem of ``winner's curse'' has been noted before \citep{zollner2007overcoming, zhong2008bias}, with some proposed solutions aiming at reducing bias (see also \cite{cohen1989two}). Other approaches to this general challenge include an empirical Bayesian approach in \cite{efron2011tweedie, wager2014geometric}, simultaneous inference methods in \cite{berk2013valid}, and  a differential privacy take on data-adaptivity in \cite{dwork2015preserving}.
One clear way out, of course, is offered by the classical concept of data-splitting, see \cite{cox}.  However in settings where the sample size is already small, reserving a hold-out data set for inference  is beyond affordability. This is often the case in human eQTL studies: for tissues other then the easily accessible blood and skin, the  relations between 20,000 genes and million of SNPs is typically studied with  a number of specimens in the hundreds at best. The same difficulties that lead to small sample sizes, make it unrealistic to simply defer the task of estimating effects to a new dataset.

In this work, we explore the potential for this problem of recent developments in the statistical literature: specifically,  the notion of conditional inference after selection and the power of randomization strategies. We offer a 
pipeline for the identification of eVariants and estimation of their effect sizes that mimics the hierarchical analysis, representing the state of the art in eQTL studies. 
Comparing the results of our pipeline with alternative strategies in simulations and real data analysis helps us understand the severity of the challenges of inference after selection in the context of eQTL. Our contribution supports the investigators in their choice of optimal use of the limited samples available in a single study, balancing the need to efficiently discover relations with that of making inference on the effect sizes of the discoveries. 

\subsection{Approach: a randomized conditional perspective}
Our methods build upon a conditional inference perspective that was introduced in \cite{exact_lasso, optimal_inference}.  The central idea is that, to counter selection bias, inference on the parameters (effect sizes, in the case of eQTL analysis) should be based on an adjusted likelihood, obtained by conditioning the data generative model upon the selection event. Conditioning has the effect of discarding the information in the data  used in selecting the eVariants, so that effect sizes are estimated on ``unused data.'' Other contributions that employ this perspective include \cite{yekutieli2012adjusted,weinstein2013selection,lee2014exact, tibshirani2016exact,selective_bayesian,reid2017post}.
In addition, we capitalize on the observation that it is possible to generalize data-splitting in a manner that allows one to make a more efficient use of the information in the sample by introducing some randomization at the selection stage \citep{randomized_response}. 

We provide a pipeline to construct confidence intervals and point estimates for the discovered effect sizes. 
The $100\cdot(1-\alpha)\%$  selection-adjusted confidence intervals are such that the probability with which each of them does not cover its target population parameter is at  most $\alpha$, conditional on selection. We will report eVariants when the selection-adjusted confidence interval for their effect size does not cover zero. This guarantee is first described in \cite{exact_lasso}, where it is called \textit{selective false-coverage} rate control. 
As a point estimator, we employ the maximum likelihood estimator (MLE) calculated from the conditional law, introduced in \cite{selective_bayesian}. We call this estimator the \textit{selection-adjusted MLE}: this serves as a quantification of the strengths of the discovered associations.

We conclude this introduction with examples of the challenges presented by selection and of how the conditional inference approach addresses them: this gives us the opportunity to introduce terminology as well as to discuss the concepts of randomization, target of inference and statistical model in this context. The rest of the paper deals with the more complicated hierarchical setting of eQTL research, which requires novel methodological results and  is organized as follows.  Section 2 describes the randomized selection pipeline that we employ as the eQTL identification strategy; section 3 presents our proposal for selection adjusted inference; section 4 contains the results of a simulation study and the concluding section 5 presents the analysis of data from GTEx. 

\subsection{Motivating examples}
\label{motivation:examples}
While the substantial contribution of the present work is to design a selection and inference pipeline that adapts to the hierarchical strategy typically adopted in eQTL studies, we start by considering a couple of ``cartoon'' examples that 
 illustrate the effect of selection bias and the overall sprit of our strategy, bypassing the complications associated to eGene selection. We focus on one gene, and imagine that the entire goal of the study is to find which of its {\em cis} variants influence its expression and with what effect sizes.  We consider two selection strategies: the first (1) consists in choosing the DNA variant that is most strongly associated with the gene (see \cite{ongen2015fast}); the second (2) uses the LASSO to identify a set of variants. In both cases, we are interested in inferring the effect sizes of the selected variants.

Introducing some notation, let $y\in \real^n$ be the response and $X\in \real^{n\times p}$ be a matrix collecting the values of predictors. Assume, without any loss of generality, that  $X$ is both centered and scaled to have columns of norm $1$. Strategy (1) identifies 
 the variant with the greatest marginal t-statistic. Let this variant correspond to the  $j_0$-th column of $X$, denoted as $X_{j_0}$ which satisfies
\[|X_{j_0}^T y /\sigma | \geq |X_{j}^T y /\sigma| \text{ for } j\in \{1,2, \cdots, p\} \setminus j_0;\]
$\sigma$ being the noise-variance in the outcome variable, here assumed known.
For strategy (2), we consider a LASSO selection with a small ridge penalty $\epsilon >0$  (for numerical stability) given by
\[\text{minimize}_{\beta} \frac{1}{2}\|y-X\beta\|_2^2 + \lambda \|\beta\|_1 + \frac{\epsilon}{2}\|\beta\|_2^2.\]
We choose the tuning parameter as $\lambda= \mathbb{E}[\|X^T \Psi\|_{\infty}], \; \Psi \sim \mathcal{N}(0,\sigma^2 I)$, a theoretical value advocated in \cite{negahban2009unified} and known to recover the true support asymptotically. 

We will be introducing randomization schemes for both strategies, which perturb the selection with the addition of some gaussian noise $\omega\in \real^p\sim \mathcal{N}(0, \tau^2 I_p)$. 
 Specifically, the randomized version of strategy (1) leads to the identification of variant $j_0$  when it satisfies:
\[|X_{j_0}^T y /\sigma + \omega_{j_0}| \geq |X_{j}^T y /\sigma + \omega_{j}| \text{ for } j\in \{1,2, \cdots, p\} \setminus j_0,\]
and the randomized version of strategy (2) solves the following  modified optimization problem   \[\text{minimize}_{\beta} \frac{1}{2}\|y-X\beta\|_2^2 -\omega^T \beta + \lambda \|\beta\|_1 + \frac{\epsilon}{2}\|\beta\|_2^2.\]

If we now consider the problem of inference following these selections, we note that we need to (a) specify a model for the data with respect to which evaluate the properties of estimators and (b) we need to formally identify the target parameters.
Of course, (a) is challenging in a context of model selection, where by definition we do not know what is the ``true'' model. Nevertheless it makes sense to work with what we might term the ``full model,'' where the mean of response variable $Y$  is parametrized by a $\real^n$ vector $\mu,$ without specifying a relation with $X$ and $Y\sim \mathcal{N}(\mu, \sigma^2 I)$: this is also the choice made in \cite{berk2013valid, exact_lasso}. 

While this full model allows us to talk precisely about the distribution of $Y,$ the target of inference (b) in both these examples  is not $\mu$, but depends on the outcome of selection. This {\em adaptive target} can be described as the projection of $\mu$ on the space spanned by the selected variables and can be interpreted as the best linear parameter which is identified through selection. For the marginal example, the adaptive target is given by $X_{j_0}^T\mu$.  Denoting $X_E$ as the selected sub-matrix, this target for the LASSO is given by $(X_E^T X_E)^{-1} X_E^T \mu$, the partial regression coefficients that are obtained by fitting a linear model to the selected set of variables $E$.

 We are now ready to explore via simulations the performance of three different inferential methods. The first (i) is  ``Vanilla'' inference that ignores selection and relies on the estimates for the adaptive targets that one would use if these were specified in advance: in the first example, this is simply the largest t-statistic, and for the second example, this is the least squared estimator with the corresponding intervals centered around these point estimates. Secondly (ii) we consider the  adjusted inference described in \cite{exact_lasso} (``Lee et al.'') which corrects for selection bias based on a screening without randomization (note that  \cite{exact_lasso} does not provide a selection-adjusted point estimate, but gives a recipe for confidence intervals). The third approach (iii) is in the spirit of the methods we develop in the rest of the paper: we condition on the outcome of randomized selection and provide adjusted confidence intervals and point estimates (``Proposed'').

In our simulations,  $X$ is fixed and a response vector $y\in \real^{n}$ is generated in each round from $Y \sim \mathcal{N}(0, I_n)$, independently of $X$. The perturbation $\omega\in \real^{p}$ is generated from $\Omega \sim \mathcal{N}(0, \tau^2 I_p); \;\tau^2 = 0.5$, independently of $Y$. We simulate $100$ such instances and compare the three inferential procedures in terms of coverage, length of the confidence intervals and risk of the point estimates. The risk metric we use is the averaged squared error deviation of the point estimates from the respective adaptive targets, described above for the two selection strategies. The results are summarized in Figure \ref{motivation}.

As expected, the unadjusted intervals based on ``Vanilla" inference fall way short of coverage. The adjusted intervals, both in the form described in \cite{exact_lasso} and in the present work, based on randomization,  achieve the target coverage. These two adjusted methods, however, differ in terms of interval length: while the ``Proposed'' intervals are $1.5$ times longer than the unadjusted intervals (a price we pay for selection), they are much shorter than the exact intervals post non-randomized screenings in \cite{exact_lasso}. This is the advantage of randomization at the selection stage: there is more ``left-over'' information available at the time of inference. The risks of the ``Vanilla" point-estimates post a non-randomized screening is also seen to be higher than the adjusted MLE, considered as a point estimate for effect sizes in this paper.
\begin{figure}[h]
\begin{center}
\centerline{\includegraphics[width=\linewidth]{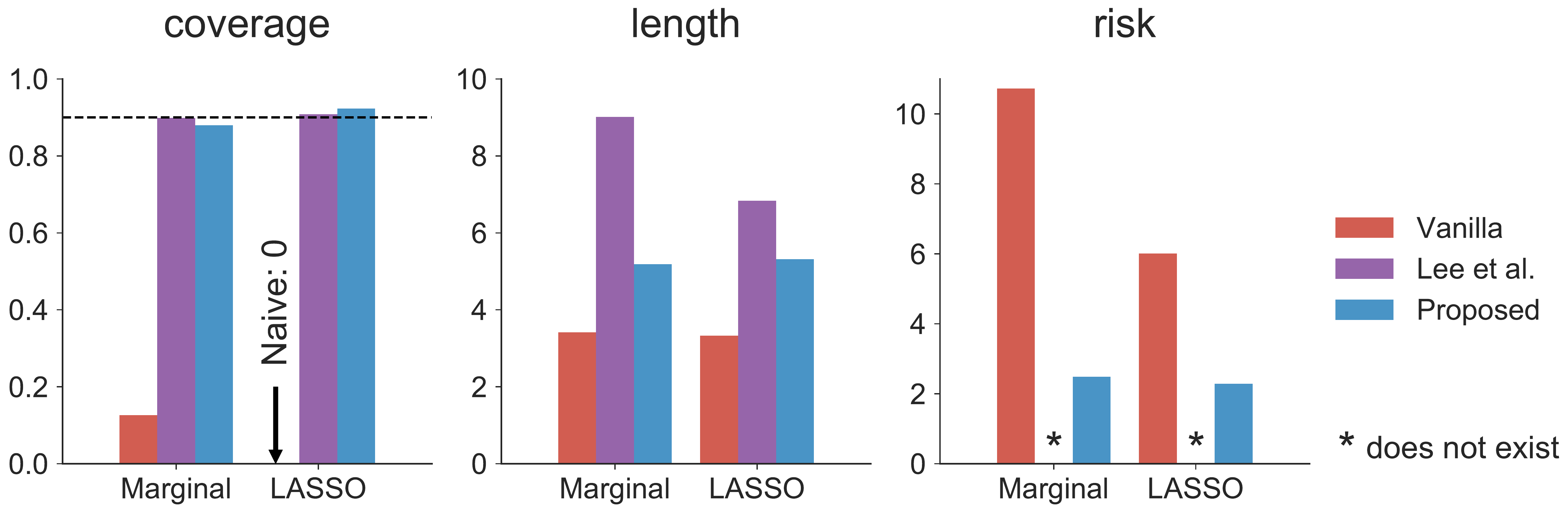}}
\caption{\small{The left most panel  compares the coverages of the confidence intervals, where the dotted black line represents the target $90\%$ coverage. The central panel in the plot gives an indication of the power of statistical inference through lengths of these intervals. Finally, the right most panel compares the empirical risks of the adjusted and unadjusted point estimates with respect to a squared error loss. 
On the $x$ axes, ``MARGINAL'' denotes the inference following selection strategy (1):  a marginal screening analysis based on a marginal t-statistic; and ``LASSO''  indicates the inferential results following selection strategy (2):  screenings by LASSO. The performances of the three inferential strategies compared are indicated using different colors, as in the legend. 
}}
\label{motivation}
\end{center}
\vskip -0.4in
\end{figure} 
In the interest of clarity, we note that while the comparisons in Figure  \ref{motivation} are meaningful on average, the selections for the non randomized and randomized versions of strategies (1) and (2) are not the same,  so that while ``Vanilla" and ``Lee'' are working with the same targets of inference in each realization of the simulations, ``Proposed'' might have slightly different targets.

In the rest of the paper, we will develop a pipeline for the analysis of eQTL data that takes advantage of randomization in the selection step to preserve more information for the inferential stage.

\section{Randomized hierarchical screening to identify interesting eQTL effects}
\label{eQTL:selection}

We start by introducing the notation that we will be using throughout. 
Our data will include measurements on the expression levels of $G$ genes and $V$ variants in $n$ subjects.  We denote the outcome variables using $Y$: $Y^{(g)} \in \real^n$  is the  vector of gene-expression levels for gene  $g\in \{1,2, \cdots ,G\}$. Let  $V_g$ be the number of variants measured within a defined {\it cis}-window around gene $g$: we indicate the matrix of the corresponding  potential explanatory variables with $X^{(g)} \in \real^{n \times V_g}$. We  assume, without  loss of generality, that $X^{(g)}$ is centered and the columns are scaled to have norm $1$. Denote the variance in gene-expression response $Y^{(g)}$ by $\sigma^2$. 
In this section we are going to outline a randomized selection procedure that mimics the hierarchical strategy often adopted in the literature \citep{GTEx17} and specify the derived ``target of inference.'' 

\subsection{A hierarchical randomized selection}
\label{hierarchical:sel}
Following the practice in eQTL studies, we want to first identify a set ${\cal G}$ of genes that appear likely to be {\em cis}-regulated, and then, for each of the genes $g\in {\cal G}$ identify a set of potential e-variants $E^{(g)}$ belonging to the {\em cis} region and appearing to have an effect on the expression of  gene $g$.
In order to both use the entire sample $n$ to guide our selection and, at the same time, reserve enough information for the following inference on the effect size, we explore the potential of a selection strategy that includes randomization in both these stages. The ``cartoon'' examples with which we concluded the previous section illustrate the advantages of randomization. For a more comprehensive discussion, please see \cite{randomized_response} and  \cite{dwork2015preserving}, which place this into the context of data re-use.

\noindent\textbf{Stage-I: Randomized selection of eGenes} \quad To discover promising genes, we test the global nulls that the expression of gene $g$ is not influenced by any of the cis-variants $V^{(g)}$: {$\cal G$} collects the set of genes for which we reject this null controlling FDR at level q. The p-values for these global null hypotheses are calculated with randomization. \\
\noindent\textbf{Step 1}.\quad  Compute for each gene $g \in  \{1,2, \cdots ,G\}$, a univariate t-test statistic based on marginal correlation of local variant $X^{(g)}_j, j\in \{1,2,\cdots, V_g\}$ with gene expression $Y^{(g)}$ and added Gaussian randomization $\omega_j^{(g)}$. The perturbed t-statistic (z-statistic for a known $\sigma$) is given by
\begin{equation}
\label{rand:T:stat}
T^{(g)}_j = {X^{(g)}_j}^T y^{(g)}/ \sigma + \omega_j^{(g)},
\end{equation}
where $\omega_j^{(g)}$ is a realization of a  Gaussian random variable $\Omega_j^{(g)}\sim\mathcal{N}(0, \gamma^2)$ with $\gamma^2$ controlling the amount of perturbation. Further, perturbations $\Omega_j^{(g)}$ are independent for $j \in \{1,2,\cdots, V_g\}$ and also, independent across all genes $g \in \{1, 2, \cdots, G\}$.
For settings where $\sigma$ is unknown, we use a marginal estimate of $\sigma$; see Remark \ref{unknown:sigma}.\\
\noindent\textbf{Step 2}.\quad  Compute a global p-value based on Bonferroni $\tilde{p}^{(g)} = V_g  p^{(g)}_{(1)}$ 
where
\[p^{(g)}_j = 2\cdot\left(1-\Phi\left(|T^{(g)}_j|/\sqrt{1 +\gamma^2}\right)\right) \text{ for } j\in \{1,2,\cdots, V_g\},\;g\in \{1,2, \cdots ,G\}.\]
\noindent\textbf{Step 3}.\quad  We apply a Benjamini-Hochberg (BH-q) procedure at level $q$ to the global p-values 
$\{\tilde{p}^{(1)}, \tilde{p}^{(2)}, \cdots, \tilde{p}^{(G)}\}.$ Based on the rejection set of BH, we report $K_0$ eGenes with the number of rejections calculated as
\[K_0 = \max_{k\in \{1,2,\cdots, G\}} \left\{\tilde{p}_{(j)}\leq \frac{j}{G}q \text{ for at most } k \text{ many p-values}\right\}.\] This set of identified genes is denoted as $\cal G$.
The analogous non-randomized selection procedure is based on t-statistics in \eqref{rand:T:stat} with no added perturbation. \\

\noindent\textbf{Stage-II: Randomized selection of potential eVariants} \quad The second stage identifies promising variants  $E^{(g)}$ for each of the eGenes $g$ in $\mathcal{G}$. 
We use a randomized version of penalized regression \citep{elastic_net}, where the $\ell_1$ penalty induces sparsity (selection) and a small  $\ell_2$ penalty  is used to regularize the problem. The set $E^{(g)}$  is identified by solving: \begin{equation}
\label{rand:lasso}
\text{mimimize}_{\beta} \frac{1}{2}\|y^{(g)} -\tilde{X}^{(g)}\beta\|_2^2 -{\zeta^{(g)}}^T \beta +\lambda \|\beta\|_1+\frac{\epsilon}{2}\|\beta\|_2^2,
\end{equation}
where $\tilde{X}^{(g)}\in \real^{n \times p_g}, \;p_g\leq V_g$ indicate a pruned set of {\em cis} variants.
The value of $\epsilon$ is small, with the only goal of ensuring the existence of a well-defined solution, while $\lambda$  is set to be the theoretical value considered in Section \ref{motivation:examples}. In the same spirit as the randomized screening of eGenes, $\zeta^{(g)}$ is based on a Gaussian variable $Z \sim \mathcal{N}(0, \tau^2 I)$.  

Note that $E^{(g)}$ indicates the set of active variants for the problem (\ref{rand:lasso}) for gene $g$: this is the result of our selection step. However, the final results of our analysis will consist of a set  $\mathcal{Q}^{(g)}\subset E^{(g)}$, corresponding to the   subset of variants with significant selection-adjusted p-values at a chosen level of selective Type-I control $\alpha$. 
\begin{remark}
We noted that problem (\ref{rand:lasso})
is defined on pruned set of variants. The variants in the {\em cis} region of a gene can be highly correlated in the sample.  This makes it hard for any multi-regression analysis like the LASSO to recover the set of variables truly associated
and hence, pruning becomes essential; prior works \cite{hastie2000gene, reid2016sparse} have recognized these challenges. We give details of an unsupervised pruning of variants using a hierarchical clustering scheme, based on the empirical correlations between variants as a distance measure in Appendix \ref{prune}. 
\end{remark}
\begin{remark}
\label{unknown:sigma}
In real data analysis, of course, the variance of gene expressions is not known: we use a plug-in estimate. For the marginal screening step, we consider a marginal estimate of variance to compute a t-statistic for the association between variant $j$ and outcome. For the randomized LASSO screening, we use an estimate of $\sigma$ from a refitted least squared regression post the LASSO. 
\end{remark}
\subsection{Model and adaptive target}
\label{adaptive:target:model}
To proceed with inference post a hierarchical screening, we assume a full Gaussian model for gene expression, as in the motivating examples: \begin{equation}
\label{saturated:model}
Y^{(g)} = \mu^{(g)} + \epsilon^{(g)},\; \epsilon^{(g)} \sim \mathcal{N}(0, \sigma^2 I).
\end{equation}
The generative law in this framework parametrizes the Gaussian mean as $\mu^{(g)}\in \real^n, g \in \{1,2,\cdots, G\}$. In addition, we assume that errors $\epsilon^{(g)}$ are independent across genes. 
\begin{remark}
Gene expression measurements and genotype data often are affected by  confounding factors, including gender, demography, platform etc. For the purposes of this paper we assume  that the data has undergone preprocessing that eliminates the effects of hidden confounding factors and measured covariates.  This allows us to assume that the noise is approximately independent across genes. Details on how these confounding factors are regressed out from the real data we analyze is provided in  Appendix \ref{dataprocessing}.
\end{remark}

Having described a model, we define the adaptive target of inference as the projection of the model parameters onto the space spanned by $E^{(g)}$ for each gene $g$:
\begin{equation} 
\label{projected:target}
{b}_{E^{(g)}} = \left(\tilde{X}_{E^{(g)}}^T \tilde{X}_{E^{(g)}}\right)^{-1} \tilde{X}_{E^{(g)}}^T \mu \in \real^{E^{(g)}}.
\end{equation}
Denoting the $j$-th component of the adaptive partial regression coefficient ${b}_{E^{(g)}} $ as ${b}_{j; E^{(g)}} = e_j^T {b}_{E^{(g)}}$, we note that  unadjusted inference for $b_{j; E^{(g)}}$ would be based on the $j$-th least squared estimator
\begin{equation}
\label{p:LC}
\hat{b}_{j; E^{(g)}} = e_j^T \left(\tilde{X}_{E^{(g)}}^T \tilde{X}_{E^{(g)}}\right)^{-1} \tilde{X}_{E^{(g)}}^T Y^{(g)}.
\end{equation}
Under an independent Gaussian noise model, $\hat{b}_{j; E^{(g)}}$ is the MLE for ${b}_{j; E^{(g)}}$ and naive confidence intervals for the same target are centered around $\hat{b}_{j; E^{(g)}}$ with a length of $2\cdot z_{1-\alpha/2} \cdot (\tilde{X}_{E^{(g)}}^T \tilde{X}_{E^{(g)}})^{-1}_{j,j}$;  $z_{1-\alpha/2}$ being the standard normal quantile. In the next section we will describe alternative point and interval estimates for these quantities that take into account the selection we have described. 

\section{Selection-adjusted inference for eVariants effects}
\label{methods}
We outline a recipe to provide selection-adjusted inference for the adaptive partial regression coefficients in \eqref{projected:target} in a coordinate-wise manner. That is, we provide the steps to get a tractable selection adjusted law: $\mathcal{L}(\hat{b}_{j; E^{(g)}} \lvert g \in \mathcal{G}, \; \hat{E}^{(g)}= E^{(g)})$, the conditional law of the $j$-th component of the least squared estimator in \eqref{p:LC} conditioned upon the event that gene $g$ was screened marginally in Stage-I and $E^{(g)}$ was screened by randomized LASSO in Stage-II.

We start again by listing the notation and terminology we will use. We then give an outline of the steps that allow us to provide a practical selection adjustment for inference, describing them in detail in following sections. 

\begin{itemize}[leftmargin=*]
\item {\bf Data and randomization} \quad At the core of the adjustment for selection is the calculation of a conditional likelihood. Because of the randomization we introduced, the relevant likelihood includes both the original data and the outcome of the randomization. The observed data $y^{(g)}$ and randomizations $(\omega^{(g)}\in \real^{V_g}, \zeta^{(g)}\in \real^{p_g})$ are realizations of the random variables $(Y^{(g)}, \Omega^{(g)}, Z^{(g)})$ respectively; $(\Omega^{(g)}, Z^{(g)})$ represent the random variables for Gaussian randomizations used in both stages of selection. The entire data across all genes is denoted by $(\bm{y}, \bm{\omega}, \bm\zeta)$ where 
\[\bm{y} = \{y^{(g)} : g\in G\}, \; \bm{\omega} =\{\omega^{(g)}: g \in G\}, \; \bm{\zeta}= \{\zeta^{(g)}: g \in G\}.\] 
\item {\bf Parameters in model} \quad In providing coordinate-wise inference, we note that the target parameter in our case is 
\[b_{j;E^{(g)}}= c_j^T \mu; \;c_j = \tilde{X}_{E^{(g)}} (\tilde{X}_{E^{(g)}}^T \tilde{X}_{E^{(g)}})^{-1} e_j. \] Based on \eqref{saturated:model}, there are nuisance parameters that are given by
$n_{j}^{(g)}= (\mu - b_{j;E^{(g)}}c_j/\|c_j\|_2)$. In a nutshell, $n_{j}^{(g)}$ are parameters that we are not interested in inferring about while providing inference for $b_{j;E^{(g)}}$. 
\item {\bf Target and null statistic} \quad A target statistic is a statistic that can be used to make inference on the target parameter ${b}_{j;E^{(g)}}$. Specifically, we will consider the $j$-th coordinate of the least squared estimator, denoted by $\hat{b}_{j;E^{(g)}}$ as the target statistic for target parameter- ${b}_{j;E^{(g)}}$. Note that $\hat{b}_{j;E^{(g)}}$ also, serves as the target statistic for naive (unadjusted) inference for ${b}_{j;E^{(g)}}$. Except now, the law of this target statistic is no longer a Gaussian distribution centered around ${b}_{j;E^{(g)}}$: the goal of the remaining section is to provide a tractable selection adjusted law for $\hat{b}_{j;E^{(g)}}$.

The null statistic according to model \eqref{saturated:model} is a data vector that is orthogonal to ${b}_{j;E^{(g)}}$ based on decomposition
$y^{(g)}= \hat b_{j; E^{(g)}}c_j/\|c_j\|_2+ \mathcal{U}_j^{(g)}$. In model \eqref{saturated:model}, the mean of the null statistic is the nuisance parameter, defined as $n_{j}^{(g)}$. Conditioning out the null statistics in the selection-adjusted law eliminates the nuisance parameters in the model.
\item {\bf Variance of target and null statistic} \quad Finally, denote as $\sigma_{j; E^{(g)}}^2 = (\tilde{X}_{E^{(g)}}^T \tilde{X}_{E^{(g)}})^{-1}_{j,j}$, the variance of $\hat{b}_{j;E^{(g)}}$ and the variance of $\mathcal{U}_j^{(g)}$ as $\Sigma_U^{(g)}$ under \eqref{saturated:model}.
\end{itemize}

\subsection{An outline of conditional inference after randomized selection in eQTL}
 In order to successfully apply conditional inference,   we need to be able to efficiently compute with the selection-adjusted law. To achieve this goal in the eQTL inference problem that we have described, we take the following steps.\\

\label{outline:methods}
\noindent {\bf I. A selection-adjusted law across the genome} \quad The randomized selection $\hat{\mathcal{G}}(\bm{y}, \bm{\omega})$ of eGenes and of promising variants for each eGene $\hat{E}^{(g)}(y^{(g)},\zeta^{(g)})$ define a selection event 
\begin{equation}
\label{coupled:sel}
\left\{(\bm{y}, \bm{\omega}, \bm{\zeta}): \hat{\mathcal{G}}(\bm{y}, \bm{\omega}) = \mathcal{G}, \;\hat{E}^{(g)}(y^{(g)},\zeta^{(g)})= E^{(g)} \text{ for } g \in\mathcal{G}\right\}.
\end{equation}
Conditioning on this  event modifies the  law of data and randomizations as follows:
\begin{equation}\prod_{g\in \mathcal{G}}\mathcal{L}(y^{(g)}; \mu)\times\mathcal{L}(\omega^{(g)}; \gamma)\times\mathcal{L}(\zeta^{(g)}; \tau)1_{\{g \in \mathcal{G},\; \hat{E}^{(g)}= E^{(g)}\}}\label{adjusted}\end{equation}
where $\mathcal{L}(y^{(g)}; \mu)\times \mathcal{L}(\omega^{(g)}; \gamma)\times\mathcal{L}(\zeta^{(g)}; \tau)$ is the unadjusted law of $(Y^{(g)}, \Omega^{(g)}, Z^{(g)})$. 

The above conditional law results from the fact that we infer about target \eqref{projected:target} for a gene $g$ only if it has been screened in the set of eGenes and the set $E^{(g)}$ of variants are chosen by a randomized version of LASSO. We point out to our readers that the selection of gene $g$ as an eGene depends not only on the data specific to gene $g$, but also on data for all other genes. This leads to a complicated selection event, which simplifies to a great extent when we condition additionally on some more information beyond knowing the set of eGenes: $\mathcal{G}$ and the respective active variants $E^{(g)}$, chosen in the next stage. This extra information includes the signs of the t-statistics for the selected genes, the signs of Lasso coefficients, the total number of eGenes $K_0$ etc.
 
\noindent {\bf II. A simplified adjusted law across eGenes}\quad  

We obtain a simplified selection event $\left\{\hat{\mathcal{H}}(y^{(g)}, \omega^{(g)}, \zeta^{(g)}) = \mathcal{H}\right\}$ by conditioning on some additional information together with the selection of eGene $g$ and variants $E^{(g)}$ selected by LASSO. Such an event is easier to handle in the sense that it is described only in terms of data specific to gene $g$ and thereby, allows a decoupling of the joint law under \eqref{adjusted} into an adjusted law for each gene in ${\cal G}$. A full description of the conditioning event, a superset of the event $\{g \in \mathcal{G},\; \hat{E}^{(g)}= E^{(g)}\}$ is available in Section \ref{decoupling}. 

Under \eqref{saturated:model} and Gaussian randomizations, the simplified selection-adjusted density for each eGene $g$ in terms of target and null statistic $(\hat{b}_{j;E^{(g)}}, \mathcal{U}_j^{(g)})$, randomizations $(\omega^{(g)}, \zeta^{(g)})$ and an observed selection $\mathcal{H}$, is given by
\begin{align}
\label{decoupled:law}
&\exp\left(-\dfrac{(\hat{b}_{j;E^{(g)}} -b_{j;E^{(g)}})^2}{2\sigma_{j;E^{(g)}}^2}\right)\cdot\exp\left(-{\|\Sigma_U^{-1/2}(\mathcal{U}_j^{(g)}-n_{j}^{(g)})\|_2^2}\right)\nonumber \\
&\times \exp\left(-{{\|\omega^{(g)}}\|_2^2}/{2\gamma^2}\right)\times \exp\left(-{\|\zeta^{(g)}\|_2^2}/{2\tau^2}\right)\times 1_{\left\{\hat{\mathcal{H}}(y^{(g)}, \omega^{(g)}, \zeta^{(g)}) = \mathcal{H}\right\}}.
\end{align}\\

\noindent{\bf III. Selection-adjusted interval and point estimates}\quad 
Since, we are interested in the selection-adjusted law of the target statistic $\hat{b}_{j;E^{(g)}}$, we marginalize over the randomizations in the joint law in \eqref{decoupled:law} and condition on nuisance statistics $\mathcal{U}_j^{(g)}$ in order to eliminate nuisance parameters $n^{(g)}_j$. This finally, leads to a selection-adjusted density for the target statistic $\hat{b}_{j;E^{(g)}}$, proportional to
\begin{equation}
\label{marginal:law}
\exp\left(-{(\hat{b}_{j;E^{(g)}} -b_{j;E^{(g)}})^2}/{2\sigma_{j;E^{(g)}}^2}\right)\mathcal{P}_{\mathcal{H}}(\hat{b}_{j;E^{(g)}}).
\end{equation}
In this law, $\mathcal{P}_{\mathcal{H}}(t) = \mathbb{P}((\hat{b}_{j; E^{(g)}}, \Omega^{(g)}, Z^{(g)}) \in \mathcal{H}\lvert \hat{b}_{j;E^{(g)}}=t)$ is the conditional probability of selection given $\hat{b}_{j;E^{(g)}}=t$.  This is the density we will use to carry out inference on the effect sizes of eVariant.

\noindent Using the adjusted law of the target statistic to define \[T(\hat{b}_{j;E^{(g)}}; b_{j;E^{(g)}}, \sigma_{j;E^{(g)}}) = \dfrac{\displaystyle\int_{\hat{b}_{j;E^{(g)}}}^{\infty} \exp\left(-{(t -b_{j;E^{(g)}})^2}/{2\sigma_{j;E^{(g)}}^2}\right)  \mathcal{P}_{\mathcal{H}}(t) dt}{\displaystyle\int_{-\infty}^{\infty} \exp\left(-{(t -b_{j;E^{(g)}})^2}/{2\sigma_{j;E^{(g)}}^2}\right)  \mathcal{P}_{\mathcal{H}}(t)dt},\]
an {\bf adjusted (two-sided) p-value} can be computed as:
\[p(\hat{b}_{j;E^{(g)}}; b_{j;E^{(g)}}, \sigma_{j;E^{(g)}}) = 2\cdot\min (T(\hat{b}_{j;E^{(g)}}; b_{j;E^{(g)}}, \sigma_{j;E^{(g)}}), 1- T(\hat{b}_{j;E^{(g)}}; b_{j;E^{(g)}}, \sigma_{j;E^{(g)}})).\] 
{\bf Confidence intervals} for $b_{j;E^{(g)}}$ with target coverage $100(1-\alpha)\%$ are constructed as
\begin{equation}
\label{intervals}
\{b \in \real : p(\hat{b}_{j;E^{(g)}}; b, \sigma_{j;E^{(g)}}) \leq \alpha\}.
\end{equation}
And, a {\bf maximum likelihood estimator} is obtained  solving
\begin{equation}
\label{MLE}
\Scale[0.95]{\underset{{b_{j;E^{(g)}}}}{\text{minimize}}(\hat{b}_{j;E^{(g)}} -b_{j;E^{(g)}})^2/{2\sigma_{j;E^{(g)}}^2}  + \log \displaystyle\int  \exp\left(-{(t -b_{j;E^{(g)}})_2^2}/{2\sigma_{j;E^{(g)}}^2}\right)\mathcal{P}_{\mathcal{H}}(t)dt.}
\end{equation}\\

\noindent {\bf IV. Approximate and tractable inference} \quad  Even though, we derived a selection adjusted law in \eqref{marginal:law}, the term $\mathcal{P}_{\mathcal{H}}(\hat{b}_{j;E^{(g)}})$, the conditional selection probability in \eqref{marginal:law}, lacks a tractable closed form expression.
This contributes to the computational bottleneck in constructing intervals or solving a MLE problem based on the adjusted law of $\hat{b}_{j;E^{(g)}}$. To offer tractable inference based on \eqref{marginal:law}, we provide an approximation for $\mathcal{P}_{\mathcal{H}}(\cdot)$.
 Details on the approximation are given in Section \ref{approximation}. Approximate inference in the form of interval and point estimates is based on plugging $\hat{\mathcal{P}}_{\mathcal{H}}(\cdot)$ in \eqref{intervals} and \eqref{MLE}.

We now make precise the steps involved in obtaining the simplified conditional law (II) and the approximation of the marginal adjusted law (IV).

\subsection{A selection-adjusted law for $\hat b_{j; E^{(g)}}$}
\label{decoupling}
With the unadjusted law of data vector and randomizations 
\[\left\{(\hat b_{j; E^{(g)}}, \mathcal{U}_j^{(g)}, \Omega^{(g)}, Z^{(g)}); \;g\in \mathcal{G}\right\}\] based on a decomposition of the response vector $y^{(g)}$ into target statistic, $\hat b_{j; E^{(g)}}$ and null statistics, $\mathcal{U}_j^{(g)}$ (see in list of terminologies)
as our starting point, we realize that the selection event in \eqref{coupled:sel} modifies this (unadjusted) density by truncating it as 
\begin{eqnarray}
&&\prod \limits_{g \in \mathcal{G}}\exp\left(-\dfrac{(\hat{b}_{j;E^{(g)}} -b_{j;E^{(g)}})^2}{2\sigma_{j;E^{(g)}}^2}-{\|\Sigma_U^{-1/2}(\mathcal{U}_j^{(g)}-n_{j})\|_2^2}-\dfrac{{\|\omega^{(g)}}\|_2^2}{2\gamma^2}-\dfrac{\|\zeta^{(g)}\|_2^2}{2\tau^2}\right)\times \nonumber\\ 
&& \;\;\;\;\;\times 1_{\left\{g\in \hat{\mathcal{G}}(\bm{y}, \bm{\omega})=\mathcal{G}, \;\hat{E}(y^{(g)},\zeta^{(g)})= E^{(g)}\right\}}.
\label{coupled:law}
\end{eqnarray}
By conditioning on some more additional information, we obtain s a simplified selection-adjusted law for $b_{j; E^{(g)}}$ in each eGene $g$. At the core of this simplification is a decoupling of the selection event (with extra information) into an event exclusive to eGene $g$ and an event that is dependent on data from all genes that are not screened in Stage-I. This is formally described in Theorem \ref{separable:selection}. The event with extra information that we condition on is given by 
\begin{equation}
\label{our:selection:event}
\begin{aligned}
&\Big\{g \in \hat{\mathcal{G}}, \; K = K_0,\;   j_{(1)}^{(g)}=  j_0, \;T_{j_{(2)}}^{(g)}= T_0,\;\text{sign}(T^{(g)}_{j_0}) = s_{j_0},\\
&\;\;\;\;\;\hat{E}(y^{(g)},\zeta^{(g)})= E^{(g)}, \text{sign}(\hat\beta_{E^{(g)}}) = s_{E^{(g)}} \Big\}
\end{aligned}
\end{equation} 
where
\begin{itemize}[leftmargin=*]
\item $K_0$ is the number of eGenes or equivalently, the number of rejections determined by BH as
\[K_0 = \max_{k\in \{1,2,\cdots, G\}} \left\{\tilde{p}_{(j)}\leq \frac{j}{G}q \text{ for at most } k \text{ many p-values}\right\},\]
based on ordered (Bonferroni-adjusted) p-values
$\{\tilde{p}_{(1)}, \tilde{p}_{(2)}, \cdots, \tilde{p}_{(G)}\}$. 
\medskip

\item $j_{(1)}^{(g)}= j_0= \arg\min_{j} p^{(g)}_j$ for eGene $g$ is the index in the set $\{1,2,\cdots, V_g \}$ associated with the largest t-statistic in \eqref{rand:T:stat} or correspondingly, the smallest marginal p-value. 
\item $T_{j_{(2)}}^{(g)}= T_0= \max_{j\neq j_{(1)}^{(g)}} |T^{(g)}_j|$
 is the second largest t-statistic in magnitude for gene $g$.
\item $\text{sign}(T^{(g)}_{j_0})= s_{j_0}$ denotes the sign of the largest t-statistics corresponding to the index $j_0$ for each eGene. 
\item $s_{E^{(g)}}= \text{sign}(\hat \beta_{E^{(g)}})$ denotes the signs of the active coefficients from the LASSO.
\end{itemize}

\begin{theorem}
\label{separable:selection}
\emph{\textit{Separability result}}
\quad The selection event in \eqref{our:selection:event} 
decouples as 
\[\{\hat{\mathcal{H}}(y^{(g)}, \omega^{(g)}, \zeta^{(g)})= \mathcal{H}\}\cap \{\hat{\mathcal{J}}(y^{(g')}, \omega^{(g')}, g'\notin \mathcal{G}) =\mathcal{J}\} \text{ where}\]
 $\{\hat{\mathcal{H}}(y^{(g)}, \omega^{(g)}, \zeta^{(g)})= \mathcal{H}\}$ is an event based on data and randomization associated exclusively with eGene $g$ and 
$\{\hat{\mathcal{J}}(y^{(g')}, \omega^{(g')}, g'\notin \mathcal{G}) =\mathcal{J}\}$ is a function of data and randomization for genes that are not rejected by the BH procedure.
\end{theorem}
A proof for Theorem \ref{separable:selection} is included in Appendix \ref{add:cond}. A consequence of the resulting separability is that the contribution from $ \{\hat{\mathcal{J}}(y^{(g')}, \omega^{(g')}, g'\notin \mathcal{G}) =\mathcal{J}\} $ in the selection-adjusted law for $\hat b_{j; E^{(g)}}$ is only that of a constant, due to the assumed independent noise structure across genes. This results in a more simplified conditional law for $(\hat b_{j; E^{(g)}}, \mathcal{U}_j^{(g)}, \Omega^{(g)}, Z^{(g)})$ in \eqref{decoupled:law}.

Finally, we obtain \eqref{marginal:law} from the joint law in \eqref{decoupled:law} by marginalizing over the randomizations and conditioning on nuisance statistics $\mathcal{U}_j^{(g)}$ in order to eliminate nuisance parameters $n^{(g)}$, as provided by the following Lemma.
\begin{lemma}
\label{marginal:law:exp}
The marginal selection-adjusted law of $\hat{b}_{j; E^{(g)}}$, conditioned on nuisance statistics $\mathcal{U}^{(g)}_j= u^{(g)}_j$ is an exponential family with a natural parameter $b_{j; E^{(g)}}$ that is given by
\[\exp\left(-{(\hat{b}_{j;E^{(g)}} -b_{j;E^{(g)}})^2}/{2\sigma_{j;E^{(g)}}^2}\right) \cdot \mathcal{P}_{\mathcal{H}(\mathcal{U}^{(g)}_j)}(\hat{b}_{j;E^{(g)}}); \text{ with }\]
\[\mathcal{P}_{\mathcal{H}(u^{(g)}_j)}(t) = \mathbb{P}\left[(\hat{b}_{j; E^{(g)}},\Omega^{(g)}, Z^{(g)}) \in \mathcal{H}(u^{(g)}_j) \;\lvert \hat{b}_{j; E^{(g)}}=t, \;\mathcal{U}^{(g)}_j = u^{(g)}_j\right]\]
where conditional on $\mathcal{U}_j^{(g)}$, the selection-induced region in terms of the target statistic and randomizations is denoted by
$\{(\hat b_{j; E^{(g)}}, \omega^{(g)}, \zeta^{(g)})\in \mathcal{H}(\mathcal{U}_j^{(g)})\}.$
\end{lemma}

\subsection{Approximate and tractable selection-adjusted law}
\label{approximation}

A technical hurdle in providing inference about ${b}_{j;E^{(g)}}$ is posed by the intractability of function $\mathcal{P}_{\mathcal{H}}(\cdot)$ in the selection-modified exponential family of distributions. We make use of two key steps in obtaining an approximate selection-adjusted distribution based on a surrogate $\hat{\mathcal{P}}_{\mathcal{H}}(\cdot)$, as outlined below.
\begin{enumerate}[leftmargin=*]
\setlength\itemsep{0.8em}
\item \textit{K.K.T. maps for hierarchical selection}. The selection map associated with the identification of eGene $g$ is given by
\begin{equation}
\label{KKT:egene:map}
\omega^{(g)}_{j_0} =  P^{(g)}_j  \hat{b}_{j; E^{(g)}}  + \eta^{(g)} + q^{(g)}_j,
\end{equation}
where $j_0$ corresponds to the index of the largest t-statistic for eGene $g$ and $P_j^{(g)}, q_j^{(g)}$ are fixed coefficients of the linear map. The screening of eGene $g$ imposes a simple thresholding constraint on $\eta^{(g)}\in \real$ of the form 
$\eta^{(g)}>C$ where $C$ is a threshold determined by the information we conditioned upon in Section \ref{decoupling} and induces a selection region $\mathcal{C}_1$ where
\begin{equation}
\label{sel:constraint:1:main}
\Scale[0.95]{\mathcal{C}_1 = \left\{\eta: \eta \geq \max\left( \sqrt{1+\gamma^2}\cdot\Phi^{-1}\left(1- \frac{K_0}{2V_g G}q\right), |T_{0}^{(g)}|\right) \right\}}.
\end{equation}

Similarly, the selection of set of active variants $E^{(g)}$ with signs $s_{E^{(g)}}$ by randomized LASSO can be written as
\begin{equation}
\label{lasso:evariant:map}
\zeta^{(g)} =  A^{(g)}_{j}  \hat{b}_{j; E^{(g)}}  + B^{(g)}_{j} o^{(g)} + c^{(g)}_{j} 
\end{equation}
where $A^{(g)}_{j}, B^{(g)}_{j} , c^{(g)}_{j} $ are fixed matrices/ vector. Selective constraints on $o^{(g)}\in \real^{p_g}$ that are variables associated with solving the LASSO objective take the form of sign and cube restrictions, which we call $\mathcal{C}_2$ where
\begin{equation}
\label{sel:constraint:2:main}
\mathcal{C}_2 =\{o: \text{sign}(o_{E^{(g)}}) = s_{E^{(g)}}, \; \|o_{-{E^{(g)}}} \|_{\infty} \leq \lambda\}.
\end{equation}
We provide comprehensive derivations of the maps in \eqref{KKT:egene:map} and \eqref{lasso:evariant:map} associated with the hierarchical randomization strategy  and the corresponding selection regions induced by these mining strategies in \eqref{sel:constraint:1:main} and in \eqref{sel:constraint:2:main} in Appendix \ref{KKT}.

\item \textit{Change of variables formulae}. We use the K.K.T. maps in \eqref{KKT:egene:map} and \eqref{lasso:evariant:map} as change of variable formulae from $(\hat{b}_{j; E^{(g)}}, \omega^{(g)} _{j_0}, \zeta^{(g)})$ to $(\hat{b}_{j; E^{(g)}},\eta^{(g)}, o^{(g)})$. The motivation behind this is that the hierarchical selection induces complicated polyhedral geometry on $(\hat{b}_{j; E^{(g)}}, \omega^{(g)} , \zeta^{(g)})$, but the constraints on $(\eta^{(g)}, o^{(g)})$ are simple threshold or sign restrictions, thereby, simplifying the computation of $\mathcal{P}_{\mathcal{H}}(\cdot)$. 
 The selection-adjusted law of $\hat{b}_{j; E^{(g)}}$ induced by the change of variables formulae is provided in Lemma \ref{change:marginal:law:exp} in Appendix \ref{change:variables} as
\[\exp\left(-\dfrac{(\hat{b}_{j;E^{(g)}} -b_{j;E^{(g)}})^2}{2\sigma_{j;E^{(g)}}^2}\right) \cdot \mathcal{P}_{\mathcal{C}_1; \mathcal{C}_2}(\hat{b}_{j;E^{(g)}})\text{ with}\]
$$\mathcal{P}_{\mathcal{C}_1; \mathcal{C}_2}(t)=\mathbb{P}\left[\eta^{(g)}\in \mathcal{C}_1, o^{(g)} \in \mathcal{C}_2 \;\lvert \hat{b}_{j; E^{(g)}}=t\right].$$
Recall that $\mathcal{C}_1$ and $\mathcal{C}_2$ are selection-induced regions described in terms of optimization variables $\eta^{(g)}$ and $o^{(g)}$ respectively.

\item \textit{A Chernoff-based approximation for $ \mathcal{P}_{\mathcal{C}_1; \mathcal{C}_2}(.)$}. For selection regions $\mathcal{C}_1$ and $\mathcal{C}_2$, Theorem \ref{sel:prob} in Appendix \ref{approx:opt} obtains an expression for $\mathcal{P}_{\mathcal{C}_1; \mathcal{C}_2}(\cdot)$ as a function of the statistic of interest $\hat{b}_{j;E^{(g)}}$. 
Finally, an upper bound in Theorem \ref{Chernoff:bound} gives an approximation for the logarithm of $\mathcal{P}_{\mathcal{C}_1; \mathcal{C}_2}(t)$ and is computed as
\begin{equation}
\label{upper:bound:chernoff}
\log C_1(t)-\inf_{\text{sign}(o_{E}) = s_{E^{(g)}}}\left\{ \| A_{E^{(g)}, j}  t + B_{E^{(g)}, j} o_E + c_{E^{(g)}, j} \|_2^2/2\tau^2- \log C_2(o_{E}, t)\right\}
\end{equation}
where 
\[C_1(t)  = \bar\Phi(\{L+ s_{j_0}P^{(g)}_j  t + s_{j_0}q^{(g)}_j\}/\gamma) \text{ and }\]
\[\Scale[0.9]{C_2(o_{E}, t ) = \prod\limits_{k \notin E^{(g)}} \Big\{\Phi\left(\{A_{-E^{(g)},j}^{k}t + B_{-E^{(g)},j}^{k}o_{E}+ \lambda\}/\tau\right) - \Phi\left(\{A_{-E^{(g)},j}^{k} t + B_{-E^{(g)},j}^{k} o_{E} -  \lambda\} /\tau\right)\Big\}.}\]
In the above, $L$ is the lower threshold on $\eta^{(g)}$; $A_{-E^{(g)},j}^{k}$ represents the $k$-th row of $A_{-E^{(g)},j}$ and $B_{-E^{(g)},j}^{k}$ denotes the $k$-th row of $B_{-E^{(g)},j}$. 

In implementations of this scheme, an approximation for the logarithm of $\mathcal{P}_{\mathcal{C}_1; \mathcal{C}_2}(\cdot)$ is finally computed as 
\begin{equation}
 \label{approx:ref}
\log C_1(t) -\inf_{o_{E}}\left\{ \| A_{E^{(g)}, j}  t + B_{E^{(g)}, j} o_E + c_{E^{(g)}, j} \|_2^2/2\tau^2- \log C_2(o_{E}, t) + \mathcal{B}(o_E)\right\}
\end{equation}
where we use a barrier penalty function $\mathcal{B}(o_E)$ to reflect the sign constraints on $o_E$ in the Chernoff bound, \eqref{upper:bound:chernoff}. Specifically, the barrier penalty that we use is $\mathcal{B}(o) = \sum_{k=1}^{E} b(o_k)$ with
\[ b(o) = \begin{cases} 
      \log(1+ 1/s o) & \text{ if } \text{sign}(o) =s \\
      \infty & \text{otherwise.}\\
    \end{cases}
\]

\end{enumerate}

\section{Performance in a {\it cis}-eQTL  study}
\label{simulation}
In this section and the next, we examine the properties of the proposed pipeline 
by testing its performance on data collected in one eQTL study \citep{GTEx17} and comparing its results with those of two other ``benchmark'' procedures.
We first rely on simulations of the outcome variables according to a known model based on genotype data in order to evaluate a number of performance metrics and then turn to the analysis of the real data.

The dataset we analyze is the collection of liver samples in V6p of the GTEx study \citep{GTEx17}. It comprises  $97$ individuals, with genotypes for 7,207,740 variants (these variants are obtained as the output of the default imputation pipeline implemented in \cite{GTEx17}) 
 and expression quantification for $21,819$ genes. Details on data acquisition and pre-processing are provided in the Appendix Section \ref{dataprocessing}.

In studying the performance of the proposed pipeline, we compare it with (1) a slightly modified version of the analysis strategy in the original GTEx paper, followed by a naive construction of the confidence interval for the effect sizes of the identified eVariant ({\em GTEx+vanilla}); (2) a strategy that uses ``out-of-the-box'' selective inference tools and employs no randomization during the selection steps ({\em Bonf+Lasso+Lee et al.}). Specifically,
\begin{description}
\item  {\em GTEx+vanilla (G.V.)}\quad The selection of eGenes is done controlling FDR at level 0.1 with the  BH procedure applied to  p-values for the global null obtained via the  Simes' combination rule \cite{simes1986improved}. We note that this differs from the original GTEx paper in two minor ways: the p-values for the global null for eGene discoveries are obtained via Simes rather than permutations in \cite{GTEx17}, and the adopted FDR controlling procedure at this stage is  BH rather than Storey's procedure referred in the GTEx paper.

In the second stage, the eVariants are selected by an adaptive forward-backward selection:  a variable is added to the regression only if it p-value is lower than the largest p-value for global nulls in the set of discovered eGenes  (with details given in \cite{GTEx17}). 

To estimate effect size, we fit the least squared estimator on the selected model for the point estimate and use the normal quantiles to construct intervals, thereby ``naively'' ignoring the selection occurred in the two screening stages.

\item {\em Bonf+Lasso+Lee et al. (B.L.L.)}\quad 
The eGene screening is conducted in a two-stage procedure similar to our proposed pipeline, with the exception that there is no randomization in the t-statistics or in the second stage selection of eVariants. The confidence intervals for effect sizes are constructed using the adjustments for the LASSO as detailed in \cite{exact_lasso}. We note that this approach (based on exact evaluation of the selection conditional probability) only accounts for the eVariant selection the second stage, but not for eGene selection in the first stage. Further, \cite{exact_lasso} does not give a selection-adjusted point estimate, and so we report the unadjusted least squared estimator on the selected variants as a point-estimate. 
\end{description}

Note that the proposed pipeline differs from {\em G.V.} in both stages of screening: the selection of eGenes and the screening of potential eVariants. These changes (swapping the construction of Simes global p-values with Bonferroni p-values in eGene screening and swapping the adaptive forward-backward selection with a LASSO screening of variants associated with eGenes) allow an easier explanation of the mining scheme and provides a more natural way to introduce randomization. Comparing our proposed procedure with these two alternatives, therefore, enables us to study how its performance relates to that of (1) a state-of-the-art  method for the identification of eVariants that does not take into account at all the effect of selection at the inferential stage, and to that of (2) an approach to correct for selection that relies on out-of-the box tools, without fully capturing the hierarchical identification of eVariants and not capitalizing on the possible power increases due to randomization.

\subsection{Simulation study}

We start with a simulation study to explore the performance of the different procedures:
using the genotypes from the GTEx liver sample, we generate artificial gene expression values and investigate how the three approaches reconstruct the effect sizes of eVariants.

\subsubsection{Data generation}
In choosing a strategy to simulate gene expression values, we followed the following principles: (a)  the selection strategy in {\em G.V.} should lead to a number of eGene and eVariant discoveries similar to those detected in the real data in \cite{GTEx17}; (b) there should be some genes that are not true eGenes, so that it is sensible to consider ``false discoveries"; (c) there should be eGenes regulated by multiple eVariants; (d) the model should be as simple as possible, to make the interpretation of the results straightforward. After experimenting with a few models, that satisfied criteria (b)-(d), we found the following to be the one that gave results closer to those in \cite{GTEx17} and adopted it.

For each gene $g$ for which gene expression is available in the real data, and for which $V_g$ {\em cis}-variants have been genotyped, we generate a vector $Y^{(g)}\in \real^n$ of  synthetic gene expression as follows
\begin{enumerate}
\item We randomly set the number of causal variants $|\mathcal{S}^{(g)}|$  from $\{0,1,2, \cdots, 9\}$ according to the distribution in \ref{negenes}, such that approximately  a third of the  genes  contain at least one true signal. 

\begin{figure}[H]
\begin{center}
\centerline{\includegraphics[width=10cm]{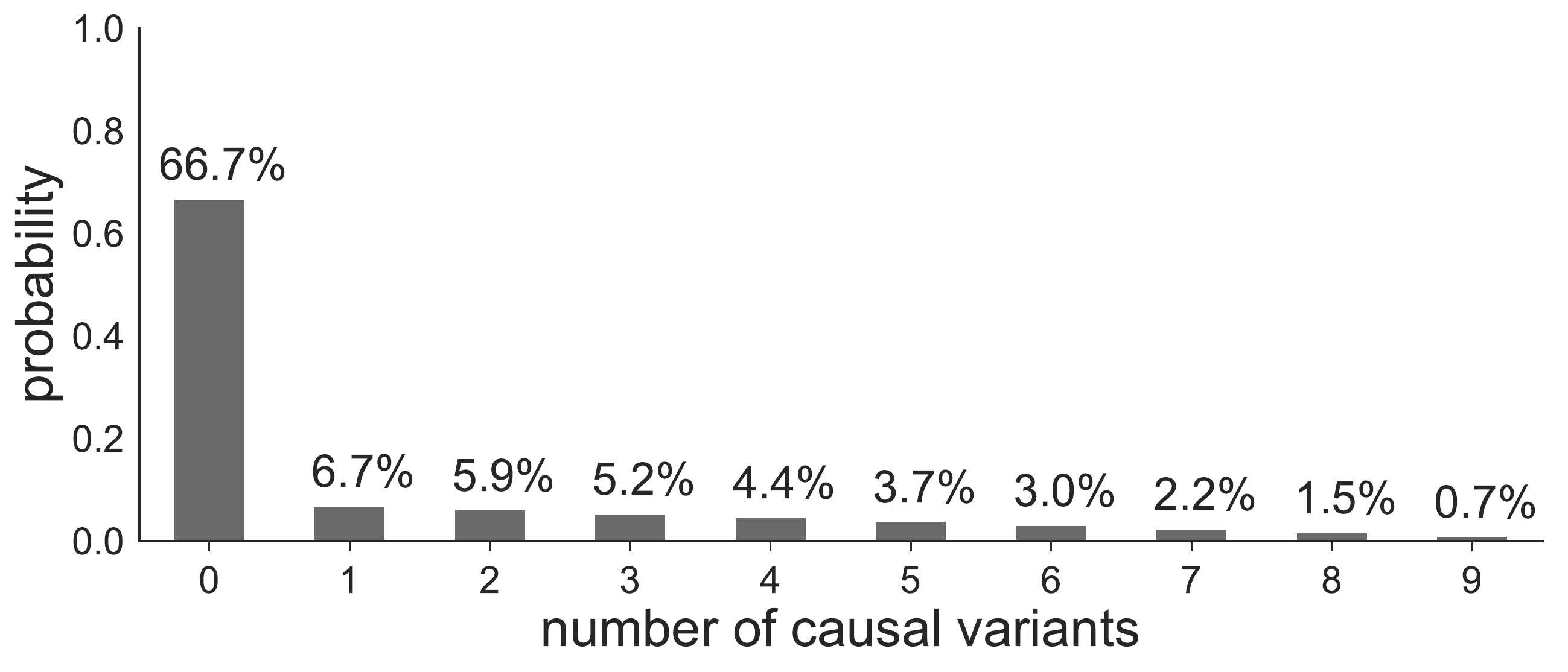}}
\caption{\small{Probability of the number of causal variants for each gene in the simulation study}}
\label{negenes}
\end{center}
\vspace{-1.cm}
\end{figure}

\item A set  $\mathcal{S}^{(g)}$ of causal variants of the selected size is  drown randomly from the variants in the {\em cis} region.  To assure, however, that the desired number of ``independent" signals is present, we want to make sure that $\mathcal{S}^{(g)}$ does not contain variants that are closely correlated and whose contribution to the gene expression value would be indistinguishable. 
To achieve this goal, we subject the genotyped variants to 
  hierarchical clustering with minimax linkage (an unsupervised pruning technique detailed in Appendix \ref{prune}): we randomly select $|\mathcal{S}^{(g)}|$ clusters,  and randomly assign an element in the cluster to be a causal variant. 
  
We also rely on this same clustering to identify a ``pruned'' set of variants that will  constitute the input of the penalized regression, which needs to work on non-collinear variables. From each of these clusters, we choose a representative, correlated by at least $\rho_0=0.5$ with all members of that cluster. We remark that the cluster representative and the possible causal variant residing in the cluster do not necessarily coincide.  
  
  \item 
  The  expression values $Y^{(g)}$ are generated according to the following model
\[Y^{(g)} = \sum_{k\in \mathcal{S}^{(g)}} {X}^{(g)}_{k} \beta_k + \epsilon^{(g)},\; \epsilon^{(g)} \sim \mathcal{N}(0, I),\]
where the noise vector $\epsilon^{(g)} \in \real ^n$ is independently sampled for each gene, each ${X}^{(g)}_{k}$ is standardized with mean zero and unit variance, and the true effect sizes are  
$\beta_k =3$. 

  \end{enumerate}

\subsubsection{Evaluation metrics}
Because of the computational costs associated to the analysis methods, we do not repeat the data generating process multiple times for a given gene, but rather we interpret our results by aggregating across genes with the same number of causal variants.

The focus of this work is to provide methods that allow correct inference for 
the adaptive target defined in \eqref{projected:target}, dependent on the selection of variants associated with each eGene.
To evaluate the performance of the three approaches with this respects, we rely on three quantities: the coverage of the confidence intervals, their lengths, and the average empirical risk of the point estimate computed with respect to a quadratic loss function. Note that the parameters to be estimated in the loss metric for risk evaluation are the adaptive targets. 

At the same time, given that the three procedures differ in their selection steps, to  interpret appropriately the results, it is useful to also compare them at the level of selection. In keeping with the hierarchical structure of the selection, there are two levels at which it makes sense to talk about FDR and power: the eGene level and the eVariant level, for selected eGenes. FDR and power are easily defined for eGenes. To  explore the performance at the level of eVariants, we are going to focus on the eVariants for selected eGenes (therefore, eVariants for erroneously missed eGenes are not going to be considered in our power evaluation). Our proposed pipeline and  {\em B.L.L.} receive as input only cluster representative SNPs and therefore can only identify these as eVariants: we consider their discoveries correct if they correspond to a cluster that contains a true causal variants;  likewise we consider a causal variant discovered if the cluster that contains it has been selected. We further note that  the selection adjusted confidence intervals can be used to refine the selection of eVariants: while regularized regression might estimate a coefficient as different from zero, if the corresponding adjusted confidence interval covers zero, it make sense to discard the eVariant from the discovery set. Therefore, it seems appropriate to evaluate FDR and power on the basis of this final post-inference set.

\subsubsection{Results}
Figures \ref{sim:inf:res} and \ref{sim:inf:res:fdr} summarize the results of the simulations, the first focusing on the inference on effect sizes, which is the object of the present paper, and the second anchoring these results in the context of FDR and Power.

Figure \ref{sim:inf:res} clearly illustrates that our pipeline succeeds in producing confidence intervals with the correct 
 empirical coverages, and that the lengths of these confidence intervals are not unduly large.  The failure of {\em G.V.} to produce confidence intervals with the right coverage is to be expected: the interest of our results is in showing the extent of the departure. While adjusting for eVariant selection as in  {\em B.L.L.} improves coverage, the approach still falls short of the target for variances in erroneously selected eGenes and for eGenes with a small number of causal variants. Moreover, the lengths of these adjusted confidence intervals are substantially longer than the ones we propose, which are just $1.5$ times longer than the naive---this is the advantage of randomization.

\begin{figure}[H]
\begin{center}
\centerline{\includegraphics[width=0.96\linewidth]{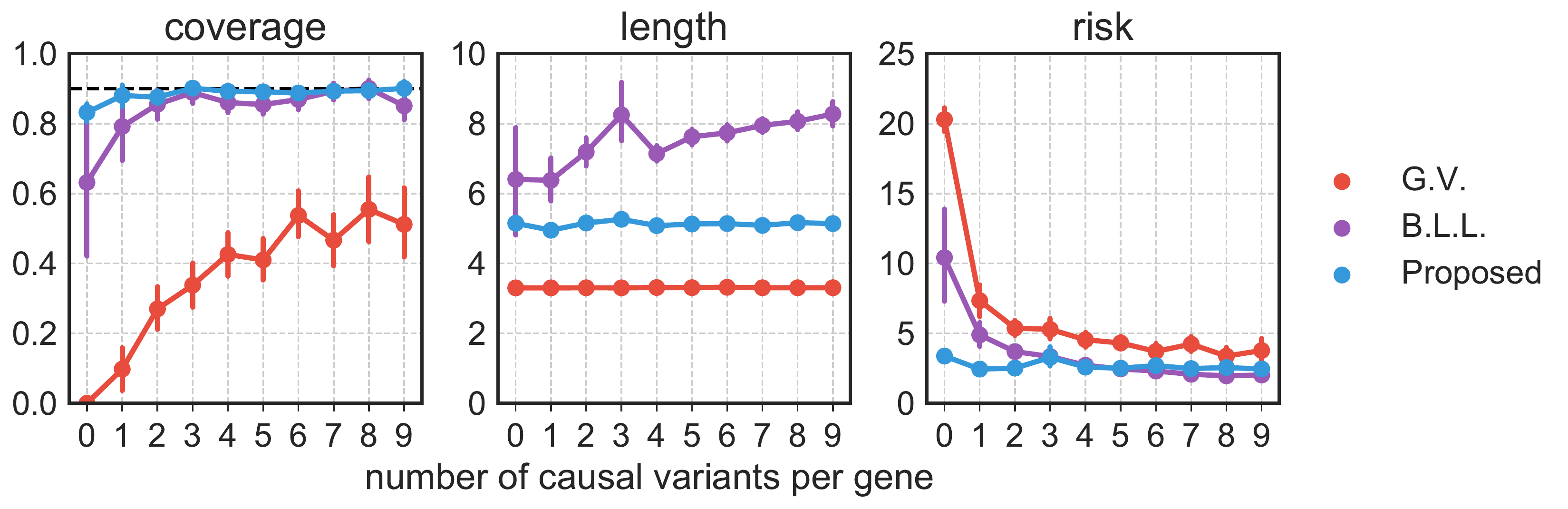}}
\caption{\small{Comparison of coverages (a)  and   lengths (b)  of confidence intervals and risk (c) of the MLE for the adaptive targets resulting from  three different methods: {\em GTEx+vanilla} (G.V., red), {\em Bonf+Lasso+Lee et al.} (B.L.L., purple) and our proposed pipeline (blue).  The values are averaged across all the selected genes with the same true number of causal variants, reported on the $x$-axis. The target   coverage is 0.9, corresponding to the  dotted horizontal line in the first panel. The dots represent the averages across all eGenes within a signal regime, and the vertical bars represent one standard deviation in both directions.}} 
\label{sim:inf:res}
\end{center}
\vspace{-0.5cm}
\end{figure}

\begin{figure}[H]
\begin{center}
\includegraphics[width=\linewidth]{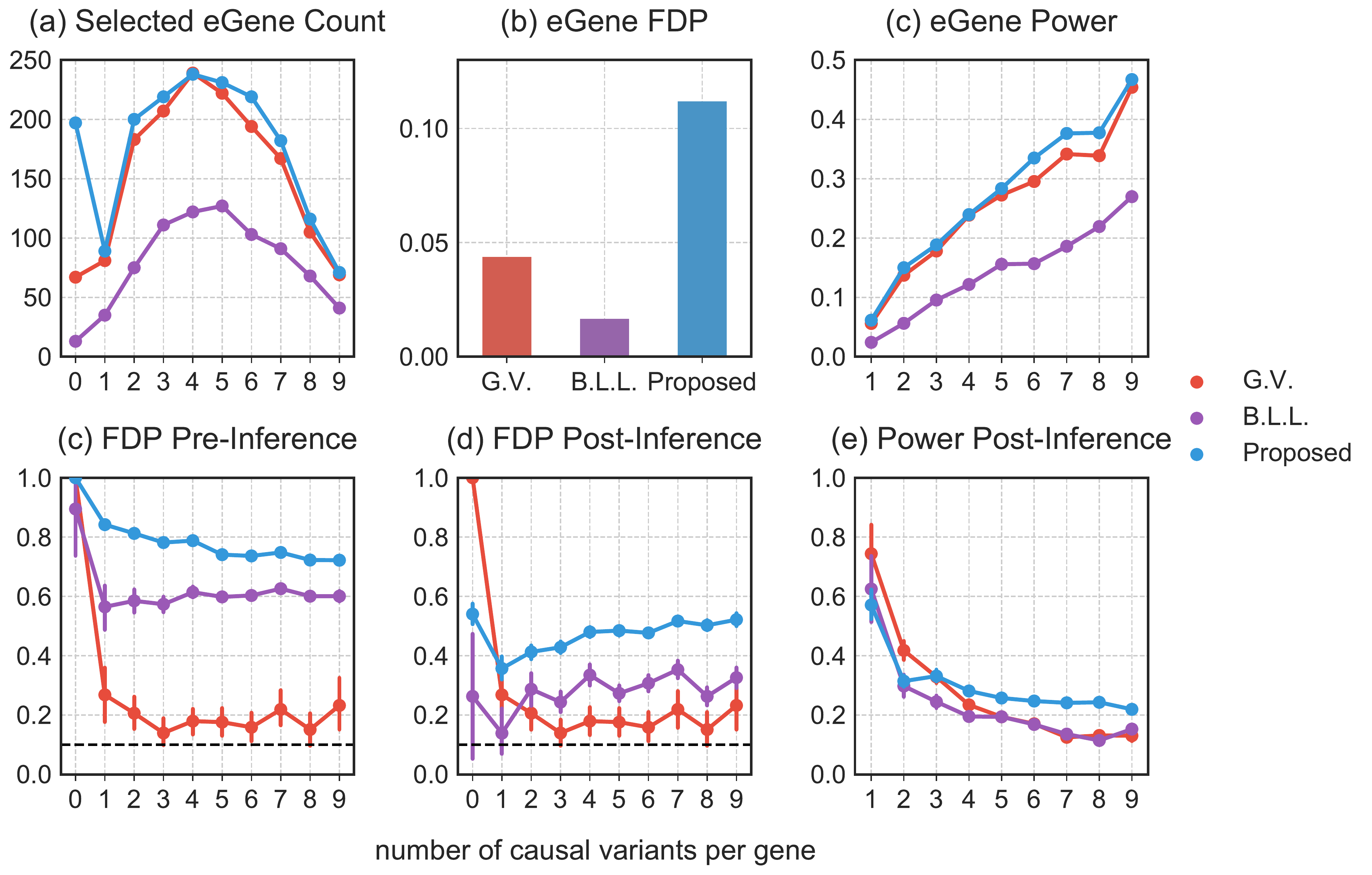}
\caption{\small{Summary of FDR (at level 0.1) and  power  for eGene and eVariants. The top row refers to eGenes and reports (a) the total number of discovered eGenes, (b) the FDR and (c) power. The bottom row, (d), (e), (f) refers to eVariants and focuses only on the eVariants relative to the discovered eGenes for each method. Panel (d)  calculates FDR for each of the gene categories on the basis of the results of the selection procedures; panel (e) shows how these results change once parameters whose adjusted confidence intervals cover zero are eliminated; and (f) illustrate the power corresponding to the selection in (e). Panels (a), (c), (d), (e), and (f) are tabulated for each of the ten categories of genes, defined by their true number of causal variants.}}
\label{sim:inf:res:fdr}
\end{center}
\vspace{-0.5cm}
\end{figure}

Figure \ref{sim:inf:res:fdr} puts these result in context of the selection outputs of the three different methods.
Looking at eGene selections, we can see that  {\em B.L.L.} is unduly conservative at the eGene level as its FDR  is much lower than the that of our procedure. It is reassuring to note that {\it Proposed} has slightly higher power than the benchmark {\it G.V.} at the eGene selection level, while controlling for FDR at the level of 0.1.

{
At the eVariant selection stage, we report the FDR of the screening procedure standard (used in {\it B.L.L.}) and randomized regularized regression (used in {\it Proposed}) and the forward-backward selection (used in {\it G.V.}). The average number of eVariants per eGenes reported by the methods are shown in Table \ref{eVariants:report:average}. Because the number of declared eVariants (post-inference) are determined by whether their intervals cover 0 or not among the screened evariants (pre-inference) for both {\it B.L.L.} and {\it Proposed}, we see the decrease in the average number of reported variants with the exception of forward-backward selection. In {\em G.V.}, the eVariants (reported variants) are selected by an adaptive forward-backward selection and this report is not impacted by inference; thereby, the FDP curve is the same both pre and post inference for {\em G.V.}.} 

Both the standard and randomized regularized regression do not appear to control the FDR  (Figure \ref{sim:inf:res:fdr} panel (d)):  correction based on the adjusted confidence intervals improves (Figure \ref{sim:inf:res:fdr} panel (e)) the performance, however still is quite inferior to the one of {\em G.V.}. This is not surprising as the adjusted approach aims at providing valid p-values, but does not apply any correction for multiplicity encountered at the eVariant level. It is rather reassuring that the selections of {\em G.V.}, which practically coincide with those reported in \cite{GTEx17},  have low FDR: this documents the efforts of a large community of scientists whose focus was precisely the selection of eVariants. 
The better selection performance of {\em G.V.} underscores the opportunity of developing alternative  adjusted approaches that condition on selection strategies as the one adopted in GTEx.

\begin{table}
\caption{The number of selected eVariants over all the selected eGenes \label{eVariants:report:average}}
\centering
\begin{tabular}{|l|l|l|l|}
\hline
  & G.V. & B.L.L. & Proposed \\ \hline
Average number of eVariants pre-inference         & 1.080  & 5.30  & 7.87     \\ \hline
Average number of eVariants post-inference        & 1.080  & 1.46 & 2.08     \\ \hline
\end{tabular}    
\end{table}

\subsection{Effect sizes in GTEx liver data }
\label{real:study}

We now turn to the analysis of the real expression data available for liver in \cite{GTEx17}. We start by noting that the limited sample size (97 observations) makes it impossible to rely on data splitting: a selection based on a 50\% of the data results in the detection of only $56$ eGenes, a mere $4\%$ of the eGenes reported in \cite{GTEx17}. We analyze the data using the three procedures we compared via simulation, reporting  the  eVariants whose adjusted-interval estimates does not cover $0$. 

The Venn-diagram \ref{venn:egenes} compares the  eGenes selected by the three procedures: our pipeline identified $2216$ eGenes, with $1683$ in common with the {\em G.V.}, which detects $1831$ eGenes in all. The {\em B.L.L.} selects $1395$ eGenes,  $1341$ of which are also identified by our proposal.
\begin{figure}
\begin{center}
\centerline{\includegraphics[width=0.5\linewidth]{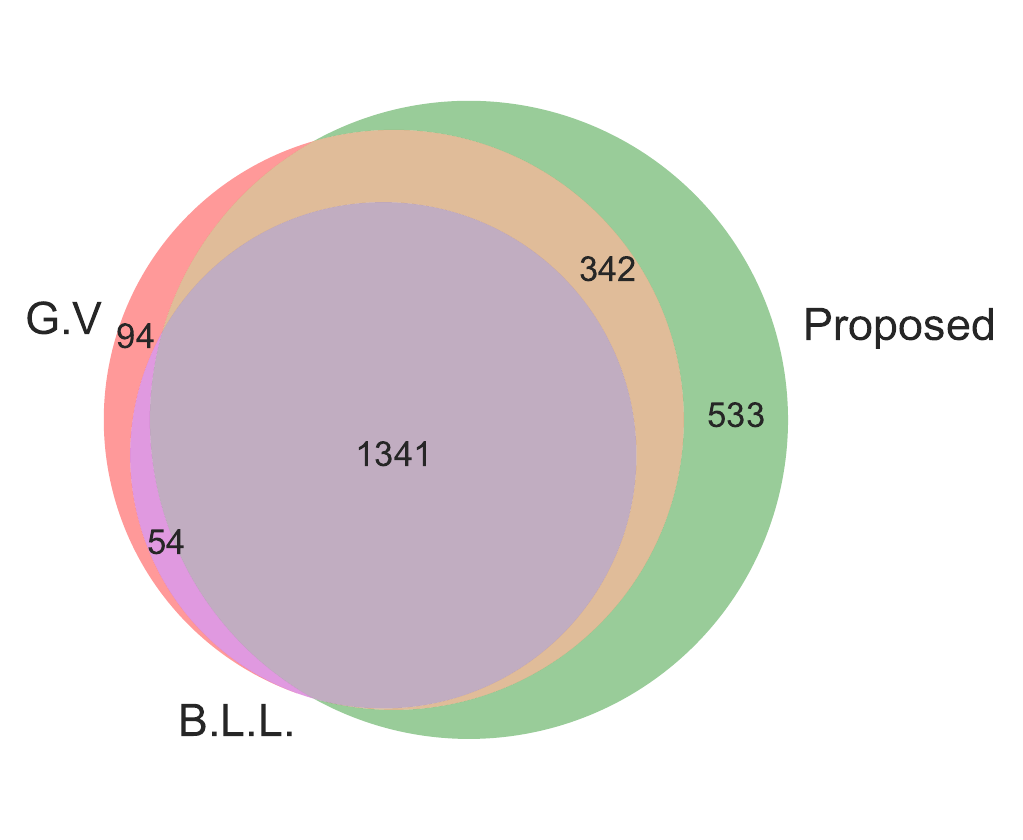}}
\caption{\small{The Venn diagram of the number of eGenes identified by each methods. Green indicates the selected eGenes unique to Proposed; pink indicates the selected eGenes unique to {\it GTEx+Vanilla} (G.V.); yellow indicates the eGenes selected by only Proposed and G.V.; magenta indicates the eGenes selected by only G.V and  {\it Bonf+Lasso+Lee et al.} (B.L.L.); dark magenta indicates the eGenes selected by all three methods. Note that there are no eGenes selected by only B.L.L..
}}

\label{venn:egenes}
\end{center}
\end{figure}

Figure \ref{comparison:evariants} illustrates the eVariants results: the number of detected eVariants and distribution of the  lengths of their confidence intervals. The pipeline we have constructed detects an average of  $4$-$5$ eVariants  per gene, while {\em B.L.L.} reports $2$-$3$ eVariants on an average and {\em G.V.} $1$ eVariant on average per eGene. Our simulations indicate that both our proposed procedure and {\em B.L.L.}  tend to have higher FDR than {\em G.V.}, so that we cannot assume all the additional discoveries to be valid ones. Nevertheless, even assuming a false discovery rate as high as 50\%, the  number of extra discoveries is such that we can expect that {\em B.L.L.} and our proposed method do allow the identification of a considerable number of additional true discoveries, or, in other words, that  {\em G.V} underestimates the number of eVariants. 
 \begin{figure}
\begin{center}
\centerline{\includegraphics[width=1.0\linewidth]{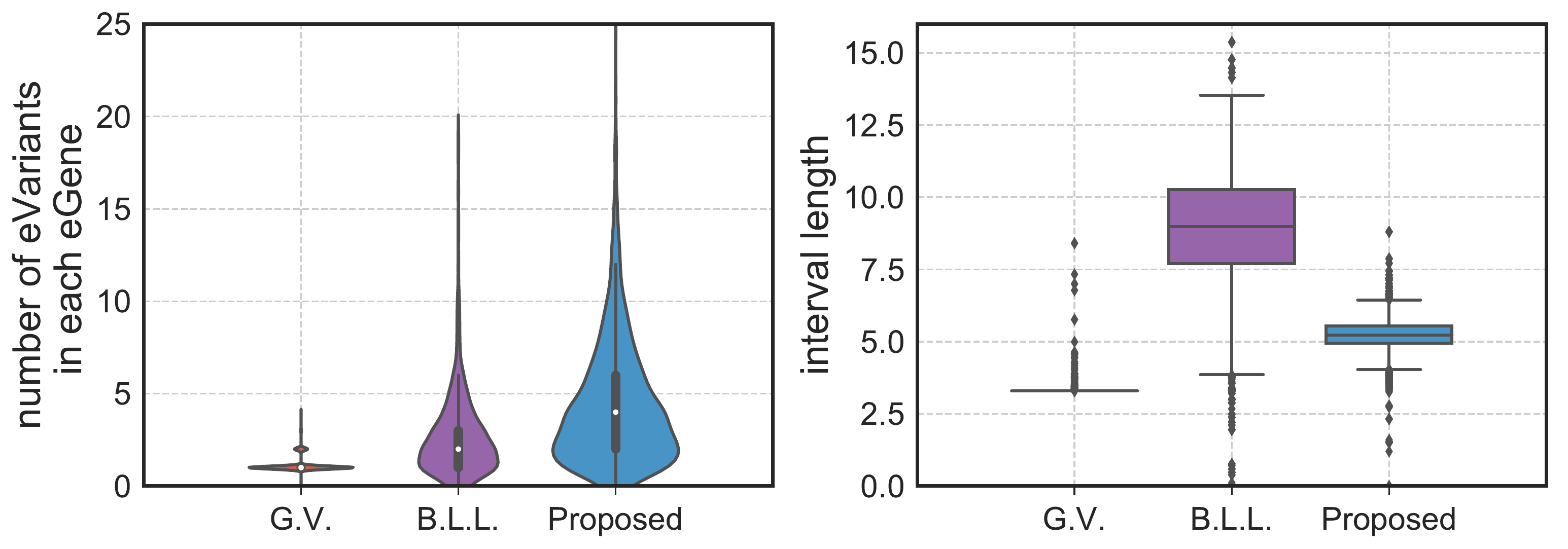}}
\caption{\small{Left panel gives the distribution of the number of eVariants per eGene as reported by the three methods of inference. Right panel plots average lengths of intervals using the interval estimates based on the same.
}
}
\label{comparison:evariants}
\end{center}
\vspace{-1.cm}
\end{figure} 

Table 2 compares the eVariant findings between {\em G.V.} and the proposed pipeline in the current work. The table in the left panel gives the number of common and exclusive eVariants reported by each procedure across the common eGenes; a total of $1683$ eGenes are reported by both. The table in the right panel compares the eVariant findings across all eGenes reported by any of these proposals. These reported findings show that even, after a deduction of $50\%$ of the discoveries from our proposal, we report potentially new discoveries that are not discovered by the benchmark analysis in {\em G.V.}.

\begin{table}
\label{eVariants:report:final}
    \caption{Comparison of total number of reported eVariants  between two methods}
    \begin{minipage}{.5\linewidth}
    	\centering
\begin{tabular}{cM{1cm}M{1cm}M{1cm}M{1cm}M{1cm}N}
			  \multicolumn{5}{c}{{\bf Across common eGenes}}            &                     \\[5ex]
                          &                            & \multicolumn{2}{c}{Proposed }            &             &        \\ [1ex]
                          & \multicolumn{1}{c|}{}      & \multicolumn{1}{c|}{$+$}    & \multicolumn{1}{c|}{$-$}    & Total  &\\[3ex] \cline{2-5} 
\multirow{2}{*}{\rotatebox{90}{GTEx benchmark}} & \multicolumn{1}{c|}{$+$}     & \multicolumn{1}{c|}{1224} & \multicolumn{1}{c|}{439} & 1663 & \\[3ex] \cline{2-5} 
                          & \multicolumn{1}{c|}{$-$}     & \multicolumn{1}{c|}{4584}  & \multicolumn{1}{c|}{}     &    &   \\[3ex] \cline{2-5} 
                          & \multicolumn{1}{c|}{Total} & \multicolumn{1}{c|}{5808} & \multicolumn{1}{c|}{}     &   & \\[3ex]
\end{tabular}
    \end{minipage}%
    \begin{minipage}{.5\linewidth}
    	\centering
\begin{tabular}{cM{1cm}M{1cm}M{1cm}M{1cm}M{1cm}}
			  \multicolumn{5}{c}{{\bf Across all eGenes}}            &                     \\[5ex]
                          &                            & \multicolumn{2}{c}{Proposed}            &               &      \\ [1ex]
                          & \multicolumn{1}{c|}{}      & \multicolumn{1}{c|}{$+$}    & \multicolumn{1}{c|}{$-$}    & Total  &\\[3ex] \cline{2-5} 
\multirow{2}{*}{\rotatebox{90}{GTEx benchmark}} & \multicolumn{1}{c|}{$+$}     & \multicolumn{1}{c|}{1224} & \multicolumn{1}{c|}{837} & 2061 & \\[3ex] \cline{2-5} 
                          & \multicolumn{1}{c|}{$-$}     & \multicolumn{1}{c|}{8175}  & \multicolumn{1}{c|}{}     &    &   \\[3ex] \cline{2-5} 
                          & \multicolumn{1}{c|}{Total} & \multicolumn{1}{c|}{9399} & \multicolumn{1}{c|}{}     &   & \\[3ex]
\end{tabular}
\end{minipage} 
\end{table}

\section{Discussion}
In many genomic studies, the first step of data analysis identifies interesting associations between variables. Inference of the parameters governing these associations is attempted only in a second stage.  Naive estimates, that rely on standard methods using the same data that suggested the importance of these associations, do not enjoy satisfactory statistical properties. A common strategy to avoid misleading conclusions is to resort to data splitting, using one portion of the data to identify parameters of interest and another to estimate them reliably.
When the studies, however, rely on scarce biological samples, data splitting (or deferring to a new sample for inference) is not a viable option. The studies of genetic regulation of gene expression in hard-to-access human tissues are a perfect exemplification of these challenges. On the one hand, to identify the parameters of interest one sifts through the possible associations between tens of thousands of genes and hundreds of thousands of DNA polymorphisms. On the other hand, the number of available samples is in the hundreds at the best.

In the present work we have explored to which extent the conceptual framework of conditional inference after selection can be brought to bear to provide researchers with reliable tools to estimate the effect size associated to DNA variants that are discovered to be associated with variation in gene expression. Following what is the standard practice in the eQTL community, we have described a two stage process for the selection of relevant {\em cis} variants: first genes under {\em cis} regulations are identified, and then the subset of important variants in their vicinity. Recognizing both the appeal that regularized regression has in the statistical community, and the ease with which we can condition on selection events based on the results of lasso, our selection of important variants for a gene of interest is carried out with a regularized regression. Acknowledging that a complete conditioning on the selection event leaves little information in the data for inference, we have exploited randomization techniques.

Deriving appropriate conditional inference in this setting presented a number of challenges. In addition to having to deal with a two-stages selection process, calculation of interval and point estimates is more involved post randomization unlike the easy computation of intervals based on a truncated Gaussian law in \cite{exact_lasso}.
To bypass the fact that the exact selection-adjusted law lacks a closed form expression, we introduce an approximation strategy that is likely to be useful in other contexts. Further, we highlight in this work a clear pay-off in terms of inferential power associated with randomized mining strategies in comparison to \cite{exact_lasso} despite incurring some cost involved in making inference in the randomized framework tractable.

A simulation study carried out starting from real data underscored the dangers of ignoring selection at the inferential stage: naive confidence intervals constructed for the effect sizes of the variants identified with a published strategy missed the target coverage by a wide margin. The pipeline we developed leads to confidence intervals with correct coverage and with lengths that are appreciably shorter than those obtained with out-of-the-box tools for conditional inference. While this results are encouraging, we also noted that the selection procedure  in our pipeline has appreciably worse performance that state-of-the-art variant selection strategies in terms of FDR. Given that we resorted to randomization to reserve information for inference, we expect a loss of precision at the selection stage. However, the size of the difference we observe is likely not to be ascribable only to randomization, but to the property of the selection procedure itself. Specifically, we observe that the regularized regression we adopted does not enjoy FDR control. In fact, we realize that the techniques in this work are amenable to a more general framework of convex learning programs, with the proposed mining pipeline as a specific example in this broader class of selection schemes. It is certainly an interesting direction for future work to study if computation for conditional inference can be carried out for other selection strategies with better FDR control, as SLOPE, the knockoffs etc. 

While we have provided effect size estimates for single tissue eQTL experiments, borrowing strengths from samples across multiple tissues (see \cite{li2013empirical}) is likely to lead to enhanced power in both discoveries and subsequent inference on their effect sizes. The adjusted likelihood provided by our work post an adequate selection strategy for discoveries based on samples, combined across different tissues can be potentially incorporated into the inferential framework. Exploring the proposed pipeline for inference in other genomic applications including the GWA (genome-wide association) studies to effectively calibrate effect sizes post discoveries by optimally using the information in the available data samples remains a challenge ahead. We conclude by remarking that we hope to take on these challenges in other biological contexts with this first attempt of tractably addressing the issue of selection bias in effect size estimation in eQTL studies.

\section{Acknowledgements}
The authors acknowledge that data for the gene expression analysis, described in this manuscript was acquired using dbGaP accession number \textbf{phs000424.v6.p1}. S.P. would like to thank Jonathan Taylor for several helpful discussions regarding this project. J.Z. acknowledges support from Stanford Graduate Fellowship. C.S. acknowledges support from NSF DMS 1712800 and NIH R01MH101782.  

\bibliographystyle{chicago}
\bibliography{references.bib}

\appendix
\section{An unsupervised pruning of local variants}
\label{prune}
The distance matrix $\{d_{i,j},\; i, j \in\; 1,2,\cdots, V_g\}$ based on which we implement an unsupervised pruning of variants, measured within a cis-window around a gene $g$ in the genome, is defined as
$d_{i,j} = d(X_i, X_j)= 1- \rho(X_i, X_j)$; $\rho$ denoting the Pearson correlation between variants $X_i$ and $X_j$. To perform a clustering of variants and obtain a prototype for each cluster, we implement a hierarchical clustering with a minimax linkage, explored in \cite{bien2011hierarchical}. This is based on a distance between clusters $\mathcal{C}_1$ and $\mathcal{C}_2$ defined as 
\[d(\mathcal{C}_1, \mathcal{C}_2) = d_0 (\mathcal{C}_1, \mathcal{C}_2),\; \text{where } d_0(\mathcal{C}) = \min_{X_i\in \mathcal{C}}\max_{X_j\in \mathcal{C}}d(X_i, X_j).\] 
The prototype for each cluster $\mathcal{C}$ is determined as \[\arg\min_{X_i\in \mathcal{C}}\max_{X_j\in \mathcal{C}}d(X_i, X_j).\]
To determine the number of clusters for a set of variants, we cut the minimax tree at a height $1-\rho_0$ that  yields a dataset of prototypical variants in which every variant has correlation of at least $\rho_0$ with one of the prototype SNPs. For our implementations, we use a $\rho_0= 0.5$ so that the prototypes $E^{(g)}$ selected by a LASSO analysis can be interpreted as identifying all the local variants that are correlated by $0.5$ with any of these representatives to be the set of promising functional variants. Prior to the secondary analysis in both the simulation study in \ref{simulation} and the GTEx data analysis in \ref{real:study}, we perform a pruning of cis-variants by the above prescribed scheme. In the simulations, pruning is seen to reduce the size of local variants for $1770$ selected eGenes with an average number of SNPs of $4497$ to $229$ SNPs on an average; in the eQTL experiment, the averaged pruned size of variants is $262$ for $2261$ eGenes with $5101$ variants on an average.

\section{Decoupled selection event: with extra conditioning}
\label{add:cond}
Theorem \ref{separable:selection} facilitates a decoupling of the truncated law across eGenes by considering a selection event 
\begin{equation*}
\begin{aligned}
&\Big\{g \in \hat{\mathcal{G}}, \; K = K_0,\;   j_{(1)}^{(g)}=  j_0, \;T_{j_{(2)}}^{(g)}= T_0,\;\text{sign}(T^{(g)}_{j_0}) = s_{j_0},\\
&\;\;\;\;\;\hat{E}(y^{(g)},\zeta^{(g)})= E^{(g)}, \text{sign}(\hat\beta_{E^{(g)}}) = s_{E^{(g)}} \Big\}
\end{aligned}
\end{equation*}
with additional information. The decoupled law is computationally less burdensome to handle. We now give a proof of this Lemma.
\begin{proof} \emph{Separability result in \ref{separable:selection}}
This selection event can be written as
\begin{align}
\label{egene:decouple}
&\Scale[0.95]{\;\;\Big\{g \in \hat{\mathcal{G}}, \; K = K_0,\;   j_{(1)}^{(g)}=  j_0, \; T_{j_{(2)}}^{(g)}= T_0,\;\text{sign}(T^{(g)}_{j_0}) = s_{j_0},}\nonumber\\
&\Scale[0.95]{\;\;\;\;\;\hat{E}(y^{(g)},\zeta^{(g)})= E^{(g)}, \text{sign}(\hat\beta_{E^{(g)}}) = s_{E^{(g)}} \Big\}}\nonumber\\
&\Scale[0.95]{= \Big\{\tilde{p}^{(g)} \leq \frac{K_0}{G} q ,\;   j_{(1)}^{(g)}=  j_0, \; T_{j_{(2)}}^{(g)}= T_0,\; \text{sign}(T^{(g)}_{j_0}) = s_{j_0},}\nonumber \\
&\Scale[0.95]{\;\;\;\;\;\;\;\hat{E}(y^{(g)},\zeta^{(g)})= E^{(g)} , \text{sign}(\hat\beta_{E^{(g)}}) = s_{E^{(g)}},}\nonumber\\
&\Scale[0.95]{\;\;\;\;\; \;\;\;\tilde{p}^{(g')}_{(k)} \geq \frac{(K_0 + k)}{G} q  \text{ for } g' \notin \mathcal{G}\Big\};\;\;\text{ where } \{ \tilde{p}^{(g')}_{(k)}\}  \text{ denote ordered p-values amongst non eGenes $g' \notin \mathcal{G}$.}}\nonumber\\
&=\Scale[0.95]{\Big\{s_{j_0} T_{j_0}^{(g)}\geq |T_{0}| , \;s_{j_0} T_{j_0}^{(g)}\geq \sqrt{1+\gamma^2}\cdot\Phi^{-1}\left(1- \frac{K_0}{2V_g G}q\right) }\nonumber\\
&\Scale[0.95]{\;\;\;\;\;\;\;\;\;\hat{E}(y^{(g)},\zeta^{(g)})= E^{(g)}, \text{sign}(\hat\beta_{E^{(g)}}) = s_{E^{(g)}} \Big\} \bigcap \left\{\tilde{p}^{(g')}_{(k)} \geq \frac{(K_0 + k)}{G} q  \text{ for } g' \notin \mathcal{G}\right\}}\nonumber\\
&=\Scale[0.95]{\Bigg\{s_{j_0} \left({X_{j_0}^{(g)}}^T y^{(g)}  + \omega_{j_0}^{(g)}\right) \geq \max\left(\sqrt{1+\gamma^2}\cdot\Phi^{-1}\left(1- \frac{K_0}{2V_g G}q\right), |T_{0}^{(g)}|\right), }\nonumber\\
&\Scale[0.95]{\;\;\;\;\;\;\;\;\hat{E}(y^{(g)},\zeta^{(g)})= E^{(g)}, \text{sign}(\hat\beta_{E^{(g)}}) = s_{E^{(g)}}\Bigg\}\bigcap \left\{\tilde{p}^{(g')}_{(k)} \geq \frac{(K_0 + k)}{G} q  \text{ for } g' \notin \mathcal{G}\right\}}.
\end{align}
The second step follows by noting that if $K_0$ is the number of rejections in the BH, then the rejected p-values satisfy
\[\tilde{p}^{(g)} \leq \frac{K_0}{G} q;\]
and the genes $g'$ (ordered in terms of their Bonferroni-adjusted p-values) that are discarded at this stage have corresponding  p-values that satisfy
\[\tilde{p}^{(g')}_{(k)} \geq \frac{(K_0 + k)}{G} q. \]
The last equality describes the selection of eGenes in terms of t-statistics.\\

\noindent Define the event $\{\hat{\mathcal{H}}(y^{(g)}, \omega^{(g)}, \zeta^{(g)})= \mathcal{H}\}$ as 
\begin{equation*}
\begin{aligned}
&\Scale[0.95]{\Bigg\{s_{j_0} \left({X_{j_0}^{(g)}}^T y^{(g)}  + \omega_{j_0}\right) \geq \max\left(\sqrt{1+\gamma^2}\cdot\Phi^{-1}\left(1- \frac{K_0}{2V_g G}q\right), |T_{0}^{(g)}|\right),}\nonumber\\
&\Scale[0.95]{ \;\;\;\;\;\;\;\;\;\hat{E}(y^{(g)},\zeta^{(g)})= E^{(g)},\text{sign}(\hat\beta_{E^{(g)}}) = s_{E^{(g)}}\Bigg\},}\nonumber
\end{aligned}
\end{equation*}
clearly  a function of $(y^{(g)},\omega_{j_0}^{(g)},\zeta^{(g)})$, exclusive to eGene $g$.
The second component $\{\hat{\mathcal{J}}(y^{(g')}, \omega^{(g')}), g'\notin \mathcal{G}) =\mathcal{J}\}$ is defined as
\[\Scale[0.97]{\left\{\tilde{p}^{(g')}_{(k)} \geq \frac{(K_0 + k)}{G} q  \text{ for } g' \notin \mathcal{G}\right\},}\]
a function of $(y^{(g')}, \omega^{(g')})$ associated with genes not selected by BH. This completes the proof that the selection region under consideration takes a separable form.
\end{proof}

Conditioning on nuisance statistics and marginalizing over randomizations gives a selection-modified family of exponential laws in Lemma \ref{marginal:law:exp} . We include a proof for the same.
\begin{proof}\emph{Exponential law in \ref{marginal:law:exp}.}
With a slight abuse of notation (noting that we use the same $\mathcal{H}$ to denote the selection region), we describe the selection event in terms of data and randomization as
\[\left\{(\hat b_{j; E^{(g)}}, \mathcal{U}_j^{(g)}, \omega^{(g)}, \zeta^{(g)})\in \mathcal{H}\right\}.\]
The joint law of data and randomizations introduced in selections is proportional to
\begin{equation*}
\begin{aligned}
&{\exp\left(-\dfrac{(\hat{b}_{j;E^{(g)}} -b_{j;E^{(g)}})^2}{2\sigma_{j;E^{(g)}}^2}\right)\cdot\exp\left(-{\|\Sigma_U^{-1/2}(\mathcal{U}_j^{(g)}-n_{j})\|_2^2}\right)}\\
&{\cdot \exp\left(-\dfrac{{\|\omega^{(g)}}\|_2^2}{2\gamma^2}\right)\cdot \exp\left(-\dfrac{\|\zeta^{(g)}\|_2^2}{2\tau^2}\right)\cdot 1_{\left\{(\hat b_{j; E^{(g)}}, \mathcal{U}_j^{(g)}, \omega^{(g)}, \zeta^{(g)})\in \mathcal{H}\right\}}}
\end{aligned}
\end{equation*}
Conditioning on null-statistics $\mathcal{U}^{(g)}_j= u^{(g)}_j$, the marginal law of $\hat{b}_{j;E^{(g)}}$ based on the above joint distribution is proportional to
\begin{equation*}
\begin{aligned}
&\exp\left(-(\hat{b}_{j;E^{(g)}} -b_{j;E^{(g)}})^2/{2\sigma_{j;E^{(g)}}^2}\right) \int \exp\left(-{{\|\omega^{(g)}}\|_2^2}/{2\gamma^2}\right)\cdot \exp\left(-{\|\zeta^{(g)}\|_2^2}/{2\tau^2}\right)\cdot\\
&\;\;\;\;\;\;\;\;\;\;\;\;\;\;\;\;\;\;\;\;\;\;\;\;\;\;\;\;\;\;\;\;\;\;\;\;\;\;\;\;\;\;\;\;\;\;\;\;\;\;\;\;\;\;\;\;\;\;\;\;\;\;\;\; 1_{\{(\hat{b}_{j; E^{(g)}},\omega^{(g)}, \zeta^{(g)}) \in \mathcal{H}(u^{(g)}_j)\}} d\omega d \zeta.\\
&=\exp\left(-{(\hat{b}_{j;E^{(g)}} -b_{j;E^{(g)}})^2}/{2\sigma_{j;E^{(g)}}^2}\right) \cdot \mathbb{P}\Big[(\hat{b}_{j; E^{(g)}},\Omega^{(g)}, Z^{(g)}) \in \mathcal{H}(u^{(g)}_j) \;\lvert \hat{b}_{j; E^{(g)}}=\hat{b}_{j; E^{(g)}},\\
&\;\;\;\;\;\;\;\;\;\;\;\;\;\;\;\;\;\;\;\;\;\;\;\;\;\;\;\;\;\;\;\;\;\;\;\;\;\;\;\;\;\;\;\;\;\;\;\;\;\;\;\;\;\;\;\;\;\;\;\;\;\;\;\;\;\;\;\;\;\;\;\; \mathcal{U}^{(g)}_j= u^{(g)}_j\Big].
\end{aligned}
\end{equation*}
This gives the adjusted marginal law of $\hat{b}_{j; E^{(g)}}$, as stated above.
\end{proof}

\section{K.K.T. mapping for selection optimizations}
\label{KKT}
We make explicit the selection maps in \eqref{KKT:egene:map} and \eqref{lasso:evariant:map} that characterize the identification of a set of eGenes, denoted by $\mathcal{G}$ and the selection of active variants corresponding to each such eGene, represented by the set $E^{(g)}$. 
These maps are obtained from the K.K.T. conditions that characterize the solutions of the underlying optimizations/ selections. The K.K.T. maps are, in both the selection of eGenes and active variants with the additionally conditioned information described in \ref{add:cond}, affine maps in both the target statistic $\hat{b}_{j;E^{(g)}}$ and randomizations, $\omega^{(g)}$ and $\zeta^{(g)}$.
\begin{itemize}[leftmargin=*]
\setlength\itemsep{1.2em}
\item Selection map for eGenes: 
The selection of eGene $g$ (with additional information) in the event $\{\mathcal{H}(y^{(g)}, \omega^{(g)}, \zeta^{(g)})= \mathcal{H}\}$ is given by
\[\Bigg\{s_{j_0} \left({X_{j_0}^{(g)}}^T y^{(g)} + \omega_{j_0}\right) \geq \max\left(\sqrt{1+\gamma^2}\cdot\Phi^{-1}\left(1- \frac{K_0}{2V_g G}q\right), |T_{0}^{(g)}|\right)\Bigg\}.\]
This is equivalent to a map
\[\eta^{(g)} = s_{j_0} \left({X_{j_0}^{(g)}}^T y^{(g)}  + \omega_{j_0}^{(g)}\right)\]
where $\eta^{(g)}$ satsifies constraints 
\begin{equation}
\label{sel:constraint:1}
\Scale[0.95]{\mathcal{C}_1 = \left\{\eta: \eta \geq \max\left( \sqrt{1+\gamma^2}\cdot\Phi^{-1}\left(1- \frac{K_0}{2V_g G}q\right), |T_{0}^{(g)}|\right) \right\}}.
\end{equation}
Decomposing the data-vector ${X_{j_0}^{(g)}}^T y^{(g)}$ into the statistic of interest $\hat{b}_{j; E^{(g)}}$ and null-statistics that are orthogonal to $\hat{b}_{j; E^{(g)}}$, this map can be re-written as follows.
\begin{align}
\omega_{j_0}^{(g)} &= -{{X_{j_0}^{(g)}}^T}\tilde{X}_{E^{(g)}} (\tilde{X}_{E^{(g)}}^T \tilde{X}_{E^{(g)}})^{-1}e_j \hat{b}_{j; E^{(g)}}/\sigma_{j;E^{(g)}}^2 + s_{j_0}\eta^{(g)} \nonumber \\
&\;\;\;\;\;\;- {{X_{j_0}^{(g)}}^T} (y - \tilde{X}_{E^{(g)}} (\tilde{X}_{E^{(g)}}^T \tilde{X}_{E^{(g)}})^{-1}e_j\hat{b}_{j; E^{(g)}}/{\sigma_{j;E^{(g)}}^2})\nonumber\\
&=  P^{(g)}_j  \hat{b}_{j; E^{(g)}}  + s_{j_0}\eta^{(g)} + \mathcal{M}^{(g)}_j, \text{ where} \nonumber
\end{align}
\[\Scale[0.95]{\mathcal{M}^{(g)}_j=- {{X_{j_0}^{(g)}}^T} (y - \tilde{X}_{E^{(g)}} (\tilde{X}_{E^{(g)}}^T \tilde{X}_{E^{(g)}})^{-1}e_j\hat{b}_{j; E^{(g)}}/{\sigma_{j;E^{(g)}}^2}) ;}\]
\[ \Scale[0.95]{P^{(g)}_j  = -{{X_{j_0}^{(g)}}^T} \tilde{X}_{E^{(g)}} (\tilde{X}_{E^{(g)}}^T \tilde{X}_{E^{(g)}})^{-1}e_j }.\]
Conditioning on $\mathcal{M}^{(g)}_j$ and denoting it as $q^{(g)}_j$, the K.K.T. map associated with the selection of eGene $g$ is represented as
\begin{equation}
\label{egene:map}
\begin{aligned}
\omega_{j_0}^{(g)} &= P^{(g)}_j  \hat{b}_{j; E^{(g)}}  + s_{j_0}\eta^{(g)} + q^{(g)}_j.
\end{aligned}
\end{equation}

\item Selection map for active variants (potential eVariants): 
Fixing notations, we let $o_{E^{(g)}}$ and $o_{-E^{(g)}}$ represent the active signs and inactive sub-gradients, obtained upon solving the LASSO objective. We denote the vector $\begin{pmatrix} o_{E^{(g)}}\\  o_{-E^{(g)}}\end{pmatrix}$ as $o^{(g)}$. The K.K.T map associated with the solution of the randomized LASSO that selects a set of variants, $E^{(g)}$ with signs $s_{E^{(g)}}$ is given by
\begin{align}
\label{KKT:eVariant}
\zeta^{(g)} &=  -\bgroup
\def\arraystretch{1.5}\begin{bmatrix} \tilde{X}_{E^{(g)}}^T \tilde{X}_{E^{(g)}} &  0 \\ \tilde{X}_{-E^{(g)}}^T \tilde{X}_{E^{(g)}} & I \end{bmatrix} \egroup\begin{pmatrix}\hat{b}_E \\ \tilde{X}_{-E}^T(y - \tilde{X}_E \hat{b}_E)\end{pmatrix}+ \bgroup
\def\arraystretch{1.5}\begin{bmatrix} \tilde{X}_{E^{(g)}}^T \tilde{X}_{E^{(g)}} +\epsilon I &  0 \\ \tilde{X}_{-E^{(g)}}^T \tilde{X}_{E^{(g)}} & I \end{bmatrix} \egroup \begin{pmatrix} o_{E^{(g)}} \\ o_{-E^{(g)}}\end{pmatrix} + \begin{pmatrix} \lambda s_{E^{(g)}} \\ 0 \end{pmatrix} \nonumber \\
&=-\bgroup
\def\arraystretch{1.5}\begin{bmatrix} \tilde{X}_{E^{(g)}}^T \tilde{X}_{E^{(g)}} &  0 \\ \tilde{X}_{-E^{(g)}}^T \tilde{X}_{E^{(g)}} & I \end{bmatrix} \egroup \begin{bmatrix} (\tilde{X}_{E^{(g)}}^T \tilde{X}_{E^{(g)}})^{-1} e_j \\ 0 \end{bmatrix} \hat{b}_{j; E^{(g)}}/\sigma_{j;E^{(g)}}^2 + \bgroup
\def\arraystretch{1.5}\begin{bmatrix} \tilde{X}_{E^{(g)}}^T \tilde{X}_{E^{(g)}} +\epsilon I &  0 \\ \tilde{X}_{-E^{(g)}}^T \tilde{X}_{E^{(g)}} & I \end{bmatrix} \egroup \begin{pmatrix} o_{E^{(g)}} \\ o_{-E^{(g)}}\end{pmatrix} \nonumber\\
&\;\;\;\;\;\;+ \begin{pmatrix} \lambda s_{E^{(g)}} \\ 0 \end{pmatrix} + \mathcal{N}^{(g)}_j  \nonumber \\
&= A_{j}^{(g)}\hat{b}_{j; E^{(g)}} + B_{j}^{(g)}o^{(g)} + c_{j}^{(g)}
\end{align}
where $o^{(g)}$ satisfies the constraints
\begin{equation}
\label{sel:constraint:2}
\mathcal{C}_2 =\{o: \text{sign}(o_{E^{(g)}}) = s_{E^{(g)}}, \; \|o_{-{E^{(g)}}} \|_{\infty} \leq \lambda\}.
\end{equation}
The matrices/ vectors  $A_{j}^{(g)}, B_{j}^{(g)}, c_{j}^{(g)}$ are equal to
\[A_{j}^{(g)} = -\frac{1}{\sigma_{j;E^{(g)}}^2}\bgroup
\def\arraystretch{1.5}\begin{bmatrix} \tilde{X}_{E^{(g)}}^T \tilde{X}_{E^{(g)}} &  0 \\ \tilde{X}_{-E^{(g)}}^T \tilde{X}_{E^{(g)}} & I \end{bmatrix} \egroup\begin{bmatrix} (\tilde{X}_{E^{(g)}}^T \tilde{X}_{E^{(g)}})^{-1} e_j \\ 0 \end{bmatrix},\;B_{j}^{(g)}=   \bgroup
\def\arraystretch{1.5}\begin{bmatrix} \tilde{X}_{E^{(g)}}^T \tilde{X}_{E^{(g)}} +\epsilon I &  0 \\ \tilde{X}_{-E^{(g)}}^T \tilde{X}_{E^{(g)}} & I \end{bmatrix} \egroup \text{ and }\]
\[ c_{j}^{(g)}=\begin{pmatrix} \lambda s_{E^{(g)}} \\ 0 \end{pmatrix} + \mathcal{N}^{(g)}_j.\]
In the above,
$$\mathcal{N}^{(g)}_j= -\bgroup\def\arraystretch{1.5}\begin{bmatrix} \tilde{X}_{E^{(g)}}^T \tilde{X}_{E^{(g)}} &  0 \\ \tilde{X}_{-E^{(g)}}^T \tilde{X}_{E^{(g)}} & I \end{bmatrix} \egroup\left\{\begin{pmatrix}\hat{b}_E \\ \tilde{X}_{-E}^T(y - \tilde{X}_E \hat{b}_E)\end{pmatrix} -  \begin{bmatrix} (\tilde{X}_{E^{(g)}}^T \tilde{X}_{E^{(g)}})^{-1} e_j \\ 0 \end{bmatrix} \hat{b}_{j; E^{(g)}}/\sigma_{j;E^{(g)}}^2 \right\}$$
representing null statistics in \eqref{marginal:law} in a saturated model framework. This map also relies on an affine decomposition into the statistic $\hat{b}_{j; E^{(g)}}$ and null statistics $\mathcal{N}^{(g)}_j$ orthogonal to it.\\

Finally, a decomposition into active (non-zero coordinates) denoted by $E^{(g)}$ and the non-active coordinates that are shrunk to $0$ by the LASSO yields a K.K.T. map in terms of $(\zeta^{(g)}, \hat{b}_{j; E^{(g)}}, o^{(g)})$ as
\begin{equation}
\label{decomp:Lasso}
\bgroup
\def\arraystretch{1.5}\begin{bmatrix}\zeta_{E^(g)}\\ \zeta_{-E^(g)} \end{bmatrix}\egroup = \bgroup
\def\arraystretch{1.5}\begin{bmatrix}A_{E^{(g)},j} \hat{b}_{j; E^{(g)}}+ B_{E^{(g)},j} o_{E^{(g)}}  + c _{E^{(g)},j}\\ A_{-E^{(g)},j} + B_{-E^{(g)},j} o_{E^{(g)}}+ o_{-E^{(g)}}  \end{bmatrix}\egroup.
\end{equation}
Here, $A_{E^{(g)},j}, B_{E^{(g)},j}, c _{E^{(g)},j}$ denote the selected subset of rows corresponding to active coordinates $E^{(g)}$. Similarly, $A_{-E^{(g)},j}, B_{-E^{(g)},j}, c _{-E^{(g)},j}$ represent the subset of rows that are associated with the inactive components or coefficients shrunk to $0$ by the LASSO.
\end{itemize}

\section{Change of variables formulae}
\label{change:variables}
It is easy to see that the maps in \eqref{egene:map} and \eqref{KKT:eVariant} together with the constraints on $(\eta^{(g)}, o^{(g)})$ give rise to half spaces defined in terms of $(\hat{b}_{j; E^{(g)}}, \omega^{(g)}, \zeta^{(g)})$ as selection-induced regions. That is, they yield a polytope
\[\left\{(\hat{b}_{j;E^{(g)}}, \omega^{(g)}, \zeta^{(g)}): (\hat{b}_{j;E^{(g)}}, \omega^{(g)}, \zeta^{(g)}) \in \mathcal{H}\right\}.\]
The reference measure in the selection-adjusted law of $\hat{b}_{j;E^{(g)}}$ is based on computing
 \[\mathcal{P}_{\mathcal{H}}(t) = \mathbb{P}((\hat{b}_{j; E^{(g)}}, \Omega^{(g)}, Z^{(g)}) \in \mathcal{H}\lvert \hat{b}_{j;E^{(g)}}=t),\]
 which is the Gaussian volume of a polytope. In general, a multivariate Gaussian volume of a polyhedral region might be hard to compute. However, we can use the K.K.T. maps associated with identifying eGenes and active variants as change of variables formulae from the space of 
$(\hat{b}_{j; E^{(g)}}, \omega^{(g)}_{j_0}, \zeta^{(g)})$ to $(\hat{b}_{j; E^{(g)}}, \eta^{(g)}, o^{(g)})$. This reduces the calculation of the reference measure to volumes of orthants and cubes, that are geometrically much simpler in nature. 

\begin{lemma}
\label{change:marginal:law:exp}
\emph{\textit{Adjusted law induced by change of measure}:}
Based on change of variables formulae
\[\omega_{j_0}^{(g)} = P^{(g)}_j  \hat{b}_{j; E^{(g)}}  + s_{j_0}\eta^{(g)} + q^{(g)}_j \text{ and }\]
\[\zeta^{(g)}=A_{j}^{(g)}\hat{b}_{j; E^{(g)}} + B_{j}^{(g)}o^{(g)} + c_{j}^{(g)};\]
the marginal selection-adjusted law of $\hat{b}_{j; E^{(g)}}$, conditioned on nuisance statistics $(\mathcal{M}^{(g)}_j, \mathcal{N}^{(g)}_j)$ is proportional to
\[\Scale[0.95]{\exp\left(-\dfrac{(\hat{b}_{j;E^{(g)}} -b_{j;E^{(g)}})^2}{2\sigma_{j;E^{(g)}}^2}\right) \cdot \mathcal{P}_{\mathcal{C}_1; \mathcal{C}_2}(\hat{b}_{j;E^{(g)}})= \exp\left(-\dfrac{(\hat{b}_{j;E^{(g)}} -b_{j;E^{(g)}})^2}{2\sigma_{j;E^{(g)}}^2}\right) \cdot \mathbb{P}\left[\eta^{(g)}\in \mathcal{C}_1, o^{(g)} \in \mathcal{C}_2 \;\lvert \hat{b}_{j; E^{(g)}}\right].}\]
\end{lemma}
\begin{proof}
Conditioning on null statistics together with a change of measure induced by the described selection maps yields a joint law for $(\hat{b}_{j; E^{(g)}}, \eta^{(g)}, o^{(g)})$ as
\begin{equation}
\begin{aligned}
&|J|\cdot\exp\left(-(\hat{b}_{j;E^{(g)}} -b_{j;E^{(g)}})^2/{2\sigma_{j;E^{(g)}}^2}\right)\cdot \exp\left(-{(P^{(g)}_j  \hat{b}_{j; E^{(g)}}  + s_{j_0}\eta^{(g)} + q^{(g)}_j)^2}/{2\gamma^2}\right)\nonumber\\
&\cdot \exp\left(-{\|A_{j}^{(g)}\hat{b}_{j; E^{(g)}} + B_{j}^{(g)}o^{(g)} + c_{j}^{(g)}\|_2^2}/{2\tau^2}\right)\cdot 1_{\left\{\eta^{(g)}\in \mathcal{C}_1, o^{(g)} \in \mathcal{C}_2\right\}}
\end{aligned}
\end{equation}
with $|J|$ being the Jacobian, a constant in our case. It is trivial to see from here that a marginalization over $(\eta^{(g)}, o^{(g)})$ yields the adjusted law with a reference measure 
\[ \mathcal{P}_{\mathcal{C}_1; \mathcal{C}_2}(t)= \mathbb{P}\left[\eta^{(g)}\in \mathcal{C}_1, o^{(g)} \in \mathcal{C}_2 \;\lvert \hat{b}_{j; E^{(g)}}=t\right].\]
\end{proof}

\section{Approximating optimization for affine Gaussian volume}
\label{approx:opt}
Theorem \ref{sel:prob} gives an expression for the probability $\mathbb{P}(\eta^{(g)} \in \mathcal{C}_1,  \; o^{(g)} \in \mathcal{C}_2\lvert \hat{b}_{j; E^{(g)}}  = t) $ as a function of $t$ where $\mathcal{C}_1$ and $\mathcal{C}_2$ represent regions based on the selection constraints in \eqref{sel:constraint:1} and \eqref{sel:constraint:2}. Theorem \ref{Chernoff:bound} derives an upper bound for the same and hence, bounds from the reference measure $ \mathcal{P}_{\mathcal{C}_1; \mathcal{C}_2}(\cdot)$ for convex and compact sets 
$\mathcal{C}_1 \subset \real \text{ and } \mathcal{C}_2\subset \real^{p_g}$.

\begin{theorem}
\label{sel:prob}
\emph{\textit{Reference measure in selection-adjusted law}:}
For the selection regions $\mathcal{C}_1\subset \real$ and $\mathcal{C}_2\subset \real^{p_g}$ in \eqref{sel:constraint:1} and \eqref{sel:constraint:2},
the probability of selection
$\mathbb{P}(\eta^{(g)} \in \mathcal{C}_1,  \; o^{(g)} \in \mathcal{C}_2\lvert \hat{b}_{j; E^{(g)}}  = t) $ equals
\[K \cdot C_1(t) \cdot \int_{o_E}\exp\left(-\|A_{E^{(g)},j} t+ B_{E^{(g)},j} o_{E}  + c _{E^{(g)},j} \|_2^2 /2\tau^2\right)C_2(o_{E}, t )1_{\{\text{sign}(o_E)= s_{E^{(g)}}\}} do_E\]
where
\[C_1(t)  = \bar\Phi(\{L+ s_{j_0}P^{(g)}_j  t + s_{j_0}q^{(g)}_j\}/\gamma) \text{ and }\]
\[C_2(o_{E}, t ) = \prod_{k \notin E^{(g)}} \Big\{\Phi\left(\{A_{-E^{(g)},j}^{k}t + B_{-E^{(g)},j}^{k}o_{E}+ \lambda\}/\tau\right) - \Phi\left(\{A_{-E^{(g)},j}^{k} t + B_{-E^{(g)},j}^{k} o_{E} -  \lambda\} /\tau\right)\Big\}\]
and $K$ is some constant that does not depend on $t$;  $A_{-E^{(g)},j}^{k}$ represents the $k$-th row of $A_{-E^{(g)},j}$ and $B_{-E^{(g)},j}^{k}$ denotes the $k$-th row of $B_{-E^{(g)},j}$; $L$ is the lower threshold in the set of constraints $\mathcal{C}_1$. 
\end{theorem}

\begin{proof}
We use notations $K_1, K_2, \cdots $ to denote constants in our proof; these dissolve as constants and do not impact the computation of the reference measure as a function of $\hat{b}_{j; E^{(g)}}$.
\begin{equation*}
\begin{aligned}
\mathbb{P}(\eta^{(g)} \in \mathcal{C}_1,  \; o^{(g)} \in \mathcal{C}_2\lvert \hat{b}_{j; E^{(g)}}  = t) &= K_1\cdot\int_{o}\int_{\eta}  \exp\left(-{(P^{(g)}_j  t + s_{j_0}\eta + q^{(g)}_j)^2}/{2\gamma^2}\right)1_{\{\eta \in \mathcal{C}_1\}}d\eta\\
&\;\;\;\;\;\;\;\;\;\cdot   \exp\left(-{\|A_{j}^{(g)}t + B_{j}^{(g)}o^{(g)} + c_{j}^{(g)}\|_2^2}/{2\tau^2}\right) 1_{\{o \in \mathcal{C}_2\}}do\\
&= K_1\cdot\int_{\eta}  \exp\left(-{(P^{(g)}_j  t   + s_{j_0}\eta + q^{(g)}_j)^2}/{2\gamma^2}\right)1_{\{\eta \in \mathcal{C}_1\}}d\eta\\
&\;\;\;\;\;\;\;\cdot \int_{o}\exp\left(-{\|A_{j}^{(g)}t + B_{j}^{(g)}o + c_{j}^{(g)}\|_2^2}/{2\tau^2}\right) 1_{\{o \in \mathcal{C}_2\}}do\\
\end{aligned}
\end{equation*}
Denoting the lower threshold in the set of constraints $\mathcal{C}_1$ as $L$, we have
\begin{align} 
\label{cube:1}
&\int_{\eta}  \exp\left(-{(P^{(g)}_j  t  + s_{j_0}\eta + q^{(g)}_j)^2}/{2\gamma^2}\right)1_{\{\eta \in \mathcal{C}_1\}}d\eta \nonumber\\
&= K_2\cdot \bar\Phi(\{L+ s_{j_0}P^{(g)}_j  t+ s_{j_0}q^{(g)}_j\}/\gamma)= K_2 \cdot C_1(\hat{\beta}_{j; E^{(g)}})
\end{align}
In the computation of 
\[\int_{o}\exp\left(-{\|A_{j}^{(g)}t + B_{j}^{(g)}o^{(g)} + c_{j}^{(g)}\|_2^2}/{2\tau^2}\right) 1_{\{o \in \mathcal{C}_2\}}do,\]
we marginalize over the inactive sub-gradient variables first. That is, this integral equals
\begin{equation}
\begin{aligned}
&\int_{o_E}\exp\left(-\|A_{E^{(g)},j} t+ B_{E^{(g)},j} o_{E}  + c _{E^{(g)},j} \|_2^2 /2\tau^2\right)1_{\{\text{sign}(o_E)= s_{E^{(g)}}\}}\nonumber\\
&\;\;\cdot\int_{o_{-E}}\exp\left(-\|A_{-E^{(g)},j} t + B_{-E^{(g)},j} o_{E}+ o_{-E}\|_2^2/2\tau^2\right) 1_{\{\|o_{-E}\|_{\infty} \leq \lambda\} } do_{-E} d o_{E}\nonumber \\
&= K_3 \cdot \int_{o_E}\exp\left(-\|A_{E^{(g)},j} t + B_{E^{(g)},j} o_{E}  + c _{E^{(g)},j} \|_2^2 /2\tau^2\right)1_{\{\text{sign}(o_E)= s_{E^{(g)}}\}}C_2(o_{E}, \hat{b}_{j; E^{(g)}} ) do_E.
\end{aligned}
\end{equation}
where $C_2(o_{E}, \hat{b}_{j; E^{(g)}} )$ equals
\begin{align}
\label{cube:lasso}
\prod_{k \notin E^{(g)}} \Big\{\Phi\left(\{A_{-E^{(g)},j}^{k}t + B_{-E^{(g)},j}^{k} o_{E}+ \lambda\}/\tau\right) - \Phi\left(\{A_{-E^{(g)},j}^{k}t + B_{-E^{(g)},j} ^{k}o_{E} -  \lambda\} /\tau\right)\Big\}.
\end{align}
The above is a consequence of the separability of the inactive sub-gradient equation and separability of the cube region induced by the same. Modulo a constant $K$, we can thereby write $\mathbb{P}(\eta^{(g)} \in \mathcal{C}_1,  \; o^{(g)} \in \mathcal{C}_2\lvert \hat{b}_{j; E^{(g)}}  = t)$ as 
\[ C_1(\hat{\beta}_{j; E^{(g)}}) \cdot \int_{o_E}\exp\left(-\|A_{E^{(g)},j} \hat{b}_{j; E^{(g)}}+ B_{E^{(g)},j} o_{E}  + c _{E^{(g)},j} \|_2^2 /2\tau^2\right)C_2(o_{E}, \hat{b}_{j; E^{(g)}} )1_{\{\text{sign}(o_E)= s_{E^{(g)}}\}} do_E.\]

\end{proof}

\begin{theorem}
\label{Chernoff:bound}
\emph{\textit{An upper bound on reference}:}
For convex and compact sets $\mathcal{C}_1 \subset \real, \mathcal{C}_2\subset \real^{p_g}$, an upper bound for $\log \mathbb{P}(\eta^{(g)} \in \mathcal{C}_1,  o^{(g)} \in \mathcal{C}_2\lvert \hat{b}_{j; E^{(g)}}  = t)$ is
\begin{equation*}
\log K + \log C_1(t)-\inf_{\text{sign}(o_{E}) = s_{E^{(g)}}}\left\{ \| A_{E^{(g)}, j}  t + B_{E^{(g)}, j} o_E + c_{E^{(g)}, j} \|_2^2/2\tau^2- \log C_2(o_{E}, t).\right\} 
\end{equation*}
where $K$ is the same constant in the conditional probability in Theorem \ref{sel:prob}.
\end{theorem}

\begin{proof}
Based on Lemma \ref{sel:prob}, we have 
\begin{equation*}
\begin{aligned}
&\log\mathbb{P}(\eta^{(g)} \in \mathcal{C}_1,  \; o^{(g)} \in \mathcal{C}_2\lvert \hat{b}_{j; E^{(g)}}  = t) \\
&= \log K + \log C_1(t) + \log\mathbb{E}\left[ C_2(o_{E^{(g)}}, \hat{b}_{j; E^{(g)}})1_{\{\text{sign}(o_{E^{(g)}}) = s_E\}}\lvert \hat{b}_{j; E^{(g)}}  = t\right].
\end{aligned}
\end{equation*}
Ignoring the constant above,
\begin{equation*}
\begin{aligned}
&\log C_1(t) + \log\mathbb{E}\left[ C_2(o_{E^{(g)}}, \hat{b}_{j; E^{(g)}} )1_{\{\text{sign}(o_{E^{(g)}}) = s_E\}}\Big\lvert \hat{b}_{j; E^{(g)}}  = t\right]\\
&\leq \log C_1(t)+ \log\mathbb{E}\left[\exp\left\{\sup\limits_{\text{sign}(o_{E}) = s_E}\log C_2(o_{E}, t)-\beta^T o_{E}\right\}\exp(\beta^T o_{E^{(g)}})\Big\lvert \hat{b}_{j; E^{(g)}}  = t\right]\\
&\leq \log C_1(t)-\sup_{ \beta \in \real^{E^{(g)}}}\inf_{\text{sign}(o_{E}) = s_E}\Big\{\beta^T o_{E}-\log C_2(o_E, t) \Big\}- \log\mathbb{E}\left[\exp(\beta^T o_{E^{(g)}})\lvert \hat{b}_{j; E^{(g)}}  = t\right]\\
&= \log C_1(t)-\inf_{\text{sign}(o_{E}) = s_E} \sup_{ \beta \in \real^{E^{(g)}}}\Big\{\beta^T o_{E}- \log\mathbb{E}\left[\exp(\beta^T o_{E^{(g)}})\lvert \hat{b}_{j; E^{(g)}}  = t\right]\Big\}-\log C_2(o_E, t).
\end{aligned}
\end{equation*}
The last step follows by a minimax equality for convex and compact sets $\mathcal{C}_1, \mathcal{C}_2$.

\noindent Finally, using the fact that \[\sup\limits_{\beta \in \real^{E^{(g)}}}\left\{\beta^{T} o_{E}- \log\mathbb{E}\left[\exp(\beta o_E)\lvert \hat{b}_{j; E^{(g)}}  = t\right]\right\}\] is the conjugate of logarithm of the moment generating functions of Gaussian random variable $o_{E^{(g)}}$ and equals 
\[\| A_{E^{(g)}, j}  t + B_{E^{(g)}, j} o_E + c_{E^{(g)}, j} \|_2^2/2\tau^2,\] we have the stated bound.
\end{proof}

We use a smooth version of the above objective in \ref{Chernoff:bound} through a barrier penalty. The final optimization that we solve to approximate the reference is given by
\[\log C_1(t) -\inf_{o_{E}}\left\{ \| A_{E^{(g)}, j}  t + B_{E^{(g)}, j} o_E + c_{E^{(g)}, j} \|_2^2/2\tau^2- \log C_2(o_{E}, t) + \mathcal{B}(o_E)\right\}\]
with $\mathcal{B}(\cdot)$ being a barrier penalty that reflects the sign constraints.

Let the approximate reference measure be $\hat{\mathcal{P}}_{\mathcal{H}}(\hat{b}_{j;E^{(g)}})$. The resulting pseudo selection-adjusted law for $b_{j;E^{(g)}}$  is proportional to
\begin{equation*}
\label{approximate:law}
\exp\left(-{(\hat{b}_{j;E^{(g)}} -b_{j;E^{(g)}})^2}/{2\sigma_{j;E^{(g)}}^2}\right) \cdot \hat{\mathcal{P}}_{\mathcal{C}_1;\mathcal{C}_2 }(\hat{b}_{j;E^{(g)}}).
\end{equation*}
In fact, a pivot based on the approximate law is calculated as
\[ \hat{p}(\hat{b}_{j;E^{(g)}}; b, \sigma_{j; E^{(g)}}^2) =\dfrac{ \sum_{t\in \chi: \geq \hat{b}_{j;E^{(g)}}} \exp\left(-{(t -b_{j;E^{(g)}})^2}/{2\sigma_{j;E^{(g)}}^2}\right) \cdot \hat{\mathcal{P}}_{\mathcal{C}_1;\mathcal{C}_2 }(t)}{ \sum_{t'\in \chi} \exp\left(-{(t' -b_{j;E^{(g)}})^2}/{2\sigma_{j;E^{(g)}}^2}\right) \cdot \hat{\mathcal{P}}_{\mathcal{C}_1;\mathcal{C}_2 }(t')}\]
for $\chi$ being an appropriate grid on the real line. 
As pointed out in the intractability of the MLE problem, solving for the selection-adjusted MLE based on \eqref{approximate:law} requires the log-partition function as a function of the target parameter $b_{j;E^{(g)}}$. 
Based on an approximated partition function on a grid $\chi\subset \real$,  the MLE is the solution of the optimization
\begin{equation}
\label{approx:MLE:problem}
\underset{{b_{j;E^{(g)}} \in \real}}{\text{minimize}}(\hat{b}_{j;E^{(g)}} -b_{j;E^{(g)}})^2/{2\sigma_{j;E^{(g)}}^2}  + \log \sum_{t\in \chi} \exp\left(-{(t -b_{j;E^{(g)}})_2^2}/{2\sigma_{j;E^{(g)}}^2}\right)\cdot \hat{\mathcal{P}}_{\mathcal{C}_1;\mathcal{C}_2 }(t).
\end{equation}

\section{Data Processing for GTEx {\it cis}-eQTL Data}
\label{dataprocessing}

The variant calls, which had been pre-processed through GTEx-specific quality control and variant imputation pipelines, were downloaded from dbGaP (accession phs000424.v6.p1),. We followed the processing pipeline for {\it cis}-eQTL according to \cite{GTEx17}: (1) the search for variants was limited to $\pm$1 Mb of the transcription start site of each gene and variants were removed if they occurred in less than 10 samples of had minor allele frequency less than 0.01; (2) a total of 20 covariates, including technical confounders, population background, gender and genotyping platform, were regressed out from the genotypes (i.e., variant matrices $X^{(g)}$)) as well as phenotypes (i.e., gene expressions $y^{(g)}$). We modified the original software fastQTL (\cite{ongen2015fast}) to implement these filters and corrections and obtain the final gene-specific genotype matrices. The processed gene expression data was also downloaded from dbGaP, and we regressed out the same 20 covariates as we did for the genotype data for the real $cis$-eQTL data analysis.

\end{document}